\documentclass{lmcs}
\usepackage[utf8]{inputenc}
\pdfoutput=1

\usepackage{lastpage}
\lmcsdoi{18}{1}{2}
\lmcsheading{}{\pageref{LastPage}}{}{}%
{Jan.~08,~2021}{Jan.~07,~2022}{}

\usepackage{macros}

\begin{document}

\title[Evaluating Homomorphism-Closed Queries on Probabilistic Graphs]{The Dichotomy of Evaluating Homomorphism-Closed Queries on Probabilistic Graphs}
\titlecomment{{\lsuper*}This paper has appeared as a conference paper at ICDT 2020 with the title: ``A Dichotomy for Homomorphism-Closed Queries on Probabilistic Graphs''~\cite{AC-ICDT20}.}

\author[A.~Amarilli]{Antoine Amarilli\rsuper{a}}
\address{LTCI, Télécom Paris, Institut Polytechnique de Paris, France}
\email{antoine.amarilli@telecom-paris.fr}

\author[{\.I}.{\.I}.~Ceylan]{{\.I}sma{\.i}l {\.I}lkan Ceylan\rsuper{b}}
\address{Department of Computer Science, University of Oxford, United Kingdom}
\email{ismail.ceylan@cs.ox.ac.uk}


\begin{abstract}
  \noindent
  We study the problem of \emph{query evaluation on probabilistic graphs}, namely, tuple-independent probabilistic databases over signatures of arity two.
  We focus on the class of queries closed under homomorphisms, or, equivalently, the \emph{infinite} unions of conjunctive queries.
  Our main result states that the probabilistic query evaluation problem is \mbox{\#P-hard} for all \emph{unbounded} queries from this class.
  As \emph{bounded} queries from this class are equivalent to a union of conjunctive queries, they are already classified by the dichotomy of Dalvi and Suciu~(2012).
  Hence, our result and theirs imply a complete data complexity dichotomy, between polynomial time and \#P-hardness, on evaluating homomorphism-closed queries over probabilistic graphs.
  This dichotomy covers in particular all fragments of infinite unions of conjunctive queries over arity-two signatures, such as \emph{negation-free (disjunctive) Datalog}, \emph{regular path queries}, and a large class of \emph{ontology-mediated queries}.
  The  dichotomy also applies to a restricted case of probabilistic query evaluation called \emph{generalized model counting}, where fact probabilities must be 0, 0.5, or 1.
  We show the main result by reducing from the problem of counting the valuations of positive partitioned 2-DNF formulae, or from the source-to-target reliability problem in an undirected graph, depending on properties of minimal models for the query.
\end{abstract}

\maketitle

\section{Introduction}%
\label{sec:intro}

The management of \emph{uncertain and probabilistic data} is an important problem in many applications, e.g., automated knowledge base construction~\cite{GoogleVault,Yago2,NELL}, data integration from diverse sources, predictive and stochastic modeling, applications based on (error-prone) sensor readings, etc. To represent probabilistic data, the most basic model is that of tuple-independent \emph{probabilistic databases (\TIDs)}~\cite{Suciu-PDBs}.
In \TIDs, every fact of the database is viewed as an independent random variable, and is either kept or discarded according to some probability. Hence, a \TID induces a probability distribution over all \emph{possible worlds}, that is, all possible subsets of the database. The central inference task for \TIDs is then \emph{probabilistic query evaluation} (\PQE): given a query $Q$, compute the probability of $Q$ relative to a \TID \pdb, i.e., the total probability of the possible worlds where~$Q$ is satisfied. We write $\PQE(Q)$ to denote the problem of \PQE\ relative to a \emph{fixed} query $Q$.

Dalvi and Suciu~\cite{dalvi2012dichotomy} obtained a dichotomy for evaluating \emph{unions of conjunctive queries (UCQs)} on tuple-independent probabilistic databases.
Their dichotomy is measured in \emph{data complexity}, i.e., as a function of the input \TID and with the query being fixed.
More specifically, they have shown that, for a given UCQ $Q$, $\PQE(Q)$ is either in polynomial time, or it is \sharpP-hard.
In the terminology of Dalvi and Suciu, a UCQ $Q$ is called \emph{safe} if $\PQE(Q)$ can be computed in polynomial time, and it is called \emph{unsafe} otherwise.
This dichotomy result laid the foundation for many other studies on the complexity of probabilistic query evaluation~\cite{Amarilli-PODS16,CDV-AIJ,FiOl16,JL12,OlHu08,OlHu09,ReSu09}.

Despite this extensive research on \TIDs, there is little known about probabilistic query evaluation for monotone query languages beyond UCQs. In particular, only few results are known for languages featuring \emph{recursion}, which is an essential ingredient in many applications. For instance, it is unknown whether \PQE\ admits a dichotomy for Datalog queries, for regular path queries, or for ontology-mediated queries~\cite{Ceylan-17}. The main motivation of this paper is thus to obtain a fine-grained classification for the complexity of probabilistic query evaluation relative to these query languages.
Our focus is on a large class of queries beyond first-order: we study the queries that are \emph{closed under homomorphisms}. We denote the class of such queries by \ucqinf as they are equivalent to \emph{infinite} unions of conjunctive queries.
We distinguish between \emph{bounded} \ucqinf queries, which are logically equivalent to a UCQ, and \emph{unbounded} \ucqinf queries, which cannot be expressed as a UCQ\@.
Notably, \ucqinf captures (negation-free) disjunctive Datalog, regular path queries, and a large class of ontology-mediated queries.

Our focus in this work is on \emph{probabilistic graphs}, i.e., probabilistic databases where all relations have at most \emph{arity two}. Data models based on binary relations are quite common in knowledge representation. Knowledge graphs such as NELL~\cite{NELL}, Yago~\cite{Yago2} and Knowledge Vault~\cite{GoogleVault} are solely based on binary relations, and are widely used for tasks such as information and relation extraction~\cite{Mintz09}, rule mining~\cite{Verbeke15}, and knowledge graph completion~\cite{Bordes}.
To encode more sophisticated domain knowledge \emph{ontologies} are employed. Ontologies are prominently formulated in description logics~\cite{BCM07}, which is a family of languages, defined over unary relations~(i.e., concepts) and binary relations~(i.e., roles).
In these and similar contexts, we want to evaluate \ucqinf queries on (graph-structured) data, while taking into account the uncertainty of the data. Therefore, we study the complexity of probabilistic query evaluation on probabilistic graphs, and ask whether evaluating \ucqinf queries admits a data complexity dichotomy.

The main result of this paper is that $\PQE(Q)$ is \#P-hard for \emph{any unbounded \ucqinf query} on probabilistic graphs. Our result thus implies a dichotomy on \PQE\  for \ucqinf over such graphs: as \emph{bounded} \ucqinf queries are equivalent to UCQs, they are already classified by Dalvi and Suciu, and we show that all other \ucqinf queries are unsafe, i.e., the \PQE\ problem is \sharpP-hard for them.
Of course, it is not surprising that \emph{some} unbounded queries in \ucqinf are unsafe for similar reasons as unsafe UCQs, but the challenge is to show hardness for \emph{every} unbounded \ucqinf query:
we do this by leveraging model-theoretic properties of this query class.

The proof consists of two main parts. First, we study \ucqinf queries with a model featuring a so-called \emph{non-iterable edge}. For all such queries, we show \#P-hardness by reducing from the problem of counting the valuations of positive partitioned 2-DNF formulae~(\pptwodnf). Second, we focus on all other unbounded queries in \ucqinf, i.e., \ucqinf queries with \emph{no} model featuring such a \emph{non-iterable edge}. For these queries, we give a reduction from the source-to-target connectivity problem in an undirected graph~(\stcon). This second reduction is considerably harder and relies on a careful study of minimal models.

This paper is organized as follows. We start by discussing closely related work for probabilistic query evaluation with a particular focus on existing classification results in Section~\ref{sec:rw}. We introduce preliminaries in Section~\ref{sec:prelim}, and formally state our result in Section~\ref{sec:result}.
We prove the result in Sections~\ref{sec:pp2dnf}--\ref{sec:ustcon}. We first deal in Section~\ref{sec:pp2dnf} with the case of queries having a model with a non-iterable edge (reducing from \pptwodnf), then argue in Section~\ref{sec:findhard} that unbounded queries must have a model with a minimal tight edge, before explaining in Section~\ref{sec:ustcon} how to use this (when the edge is iterable) to reduce from \stcon. We then present two generalizations of our main result in Section~\ref{sec:generalizations}. We conclude in Section~\ref{sec:conc}.

\section{Related Work}%
\label{sec:rw}
Research on probabilistic databases is a well-established field; see e.g.~\cite{Suciu-PDBs}. The first dichotomy for queries on such databases was shown by Dalvi and Suciu~\cite{DaSu07}: a self-join-free conjunctive query is safe if it is \emph{hierarchical}, and \sharpP-hard otherwise. They then extended this result to a dichotomy for all UCQs~\cite{dalvi2012dichotomy}.
Beyond UCQs, partial dichotomy results are known for some queries with negation~\cite{FiOl16}, with disequality~($\neq$) joins  in the queries~\cite{OlHu08}, or with inequality~($<$) joins~\cite{OlHu09}. Some results are known for extended models, e.g., the dichotomy of Dalvi and Suciu has been lifted from \TIDs to open-world probabilistic databases~\cite{CDV-AIJ}. However, we are not aware of dichotomies in the probabilistic database literature that apply to Boolean queries beyond first-order logic, or to queries with fixpoints.
Query evaluation on probabilistic databases has also been studied in restricted contexts, e.g., when probabilistic tuples are only allowed to have probability $0.5$. This is for instance the focus of the recent paper of Kenig and Suciu~\cite{KenigSuciu20}, which we discuss in Section~\ref{sec:generalizations}.

Query evaluation on probabilistic graphs has also been studied in the context of \emph{ontology-mediated queries} (OMQs)~\cite{JL12,BCL-AAAI17,BCL19}. An OMQ is a composite query that typically consists of a UCQ and an \emph{ontology}. The only classification result on \PQE\ for OMQs beyond first-order-rewritable languages is given for the description logic \ELI~\cite{JL12}.
This result applies to a class of queries that go beyond first-order logic.
Our work generalizes this result (Theorem~6 of~\cite{JL12}) by showing hardness for any unbounded \ucqinf, not just the ones expressible as OMQs based on \ELI\@. Part of our techniques (Section~\ref{sec:pp2dnf}) are related to theirs, but the bulk of our proof (Sections~\ref{sec:findhard} and~\ref{sec:ustcon}) uses new techniques, the need for which had in fact been overlooked in~\cite{jung2014reasoning,JL12}.
Specifically, we identified a gap in the proofs of Theorem~6 of~\cite{JL12} and Theorem~5.31 of~\cite{jung2014reasoning} concerning a subtle issue of ``back-and-forth'' matches related to the use of inverse roles of $\ELI$. We have communicated this with the authors of~\cite{jung2014reasoning,JL12}, which they kindly acknowledged~\cite{JL20}.
Our proof thus completes the proof of Theorem~6 in~\cite{JL12}, and generalizes it to all unbounded \ucqinf.

\section{Preliminaries}%
\label{sec:prelim}
In this section, we introduce all technical preliminaries relevant to our study. In particular, we introduce the query languages studied in this paper, and the tuple-independent probabilistic database model. We also discuss briefly the complexity classes relevant to our study, as well as two canonical $\sharpP$-hard problems which are used later in the reductions.

\subsection*{Vocabulary.} We consider a \emph{relational signature} $\sigma$ which is a set of \emph{predicates}. In this work, the signature is required to be \emph{arity-two}, i.e., it consists \emph{only} of predicates of arity two. Our results can easily be extended to signatures with relations having predicates of arity one and two, as we show in Section~\ref{sec:generalizations}.

A \emph{$\sigma$-fact} is an expression of the form $F = R(a, b)$ where $R$ is a predicate and $a, b$ are constants.
By a slight abuse of terminology, we call $F$ a \emph{unary} fact if~$a = b$, and a \emph{non-unary fact}
otherwise. A \emph{$\sigma$-atom} is defined in the same way with variables instead of constants. For brevity, we will often talk about a  \emph{fact} or an \emph{atom} when~$\sigma$ is clear from context. We also speak of \emph{$R$-facts} or \emph{$R$-atoms} to specifically refer to facts or atoms that use the predicate $R$.

It will be convenient to write $\sigma^\leftrightarrow$ the arity-two signature consisting of the relations of~$\sigma$ and of the relations $R^-$ for~$R \in \sigma$, with a semantics that we define below.

\subsection*{Database instances.}
A \emph{database instance over~$\sigma$}, or a \emph{$\sigma$-instance}, is a set of facts over~$\sigma$.
All instances considered in this paper are finite. The \emph{domain} of a fact~$F$, denoted $\dom(F)$, is the set of constants that appear in~$F$, and the \emph{domain} of an instance $I$, denoted $\dom(I)$, is
the set of constants that appear in~$I$, i.e., the union of the domains of its facts.

Every $\sigma$-instance $I$  can be seen as a $\sigma^\leftrightarrow$-instance consisting of all the
$\sigma$-facts in~$I$, and all the facts $R^-(b, a)$ for each fact $R(a, b)$ of~$I$. Thus, for a $\sigma$-instance~$I$, and for an element $a \in \dom(I)$, we define the set of all $\sigma^\leftrightarrow$-facts of the form $F = R(a, b)$ in~$I$ as:
\begin{align*}
  \{ & S(a,a)\phantom{^-} \mid S\in\sigma, S(a,a) \in I\} \\
\cup~\{ & S(a,b)\phantom{^-} \mid S\in\sigma, b \in \dom(I), S(a,b) \in I\} \\
\cup~\{ & S^-(a,b) \mid S\in\sigma, b \in \dom(I), S(b,a) \in I\}.
\end{align*}
If we say that we create a fact $R(a, b)$ for $R \in \sigma^\leftrightarrow$, we
mean that we create $S(a, b)$ if $R = S$ for some $S \in \sigma$, and $S(b, a)$
if $R = S^-$ for some $S \in \sigma$.

The \emph{Gaifman graph} of an instance $I$ is the undirected graph having $\dom(I)$ as vertex set, and having an edge $\{u, v\}$ between any two $u \neq v$ in~$\dom(I)$ that co-occur in some fact of~$I$. An instance is \emph{connected} if its Gaifman graph is connected. We call $\{u, v\}$ an (undirected) \emph{edge} of~$I$, and the facts of~$I$ that it \emph{covers} are the $\sigma$-facts of~$I$ whose domain is a subset of~$\{u, v\}$.
Note that a fact of the form $R(u, u)$ is covered by all edges involving~$u$. Slightly abusing notation, we say that an \emph{ordered} pair $e = (u, v)$ is a (directed) \emph{edge} of~$I$ if $\{u, v\}$ is an edge of the Gaifman graph, and say that it \emph{covers}
the following $\sigma^\leftrightarrow$-facts of~$I$:
\begin{align*}
&\{ S(u,u)  \mid S \in \sigma, S(u,u) \in I\}\\
  \cup~&\{ S(v,v)  \mid S \in \sigma, S(v,v) \in I\} \\
  \cup~&\{ S(u,v)  \mid  S \in \sigma, S(u,v) \in I\} \\
  \cup~&\{ S^-(u,v) \mid  S \in \sigma, S(v,u) \in I\}.
\end{align*}
Note that the directed edge $(v, u)$ covers the same facts as~$(u, v)$, except that in non-unary facts the relations $S\in\sigma$ and the reverse relations $S^-$ are swapped.

In the course of our proofs, we will often modify instances in a specific way, which we call \emph{copying} an edge.
Let $I$ be an instance, let $(u, v)$ be a directed edge of~$I$, and let $u', v'$ be any elements of~$\dom(I)$. If we say that we \emph{copy} the edge~$e$ on~$(u', v')$, it means that we modify~$I$ to add a copy of each fact covered by the edge~$e$, but using $u'$ and $v'$ instead of~$u$ and~$v$.
Specifically, we create $S(u', v')$ for all $\sigma$-facts of the form $S(u, v)$ in~$I$, we create $S(v', u')$ for all $\sigma$-facts of the form $S(v, u)$ in~$I$, and we create $S(u',u')$ and $S(v', v')$ for all $\sigma$-facts respectively of the form $S(u, u)$ and $S(v,v)$ in~$I$. Of course, if some of these facts already exist, they are not created again. Note that $(u', v')$ is an edge of~$I$ after this process.

An instance $I$ is a \emph{subinstance} of another instance $I'$ if $I \subseteq I'$, and $I$ is a \emph{proper subinstance} of~$I'$ if $I \subsetneq I'$. Given a set $S \subseteq \dom(I)$ of domain elements, the subinstance of~$I$ \emph{induced} by~$S$ is the instance formed of all the facts $F \in I$ such that $\dom(F) \subseteq S$.

A \emph{homomorphism} from an instance $I$ to an instance~$I'$ is a function $h$ from $\dom(I)$ to~$\dom(I')$ such that, for every fact $R(a, b)$ of~$I$, the fact $R(h(a), h(b))$ is a fact of~$I'$. In particular, whenever $I \subseteq I'$ then $I$ has a homomorphism to~$I'$. An \emph{isomorphism} is a bijective homomorphism whose inverse is also a homomorphism.

\subsection*{Query languages.}
Throughout this work, we focus on Boolean queries. A (Boolean) \emph{query} over a signature~$\sigma$ is
a function from $\sigma$-instances to Booleans. An instance $I$ \emph{satisfies} a query~$Q$ (or $Q$ \emph{holds} on~$I$, or $I$ is a \emph{model} of~$Q$), written $I \models Q$, if $Q$ returns true when applied to~$I$; otherwise, $I$ \emph{violates} $Q$. We say that two queries $Q_1$ and $Q_2$ are \emph{equivalent} if for any instance~$I$, we have $I \models Q_1$ iff $I \models Q_2$.
In this work, we study the class \ucqinf of queries that are \emph{closed under homomorphisms} (also called \emph{homomorphism-closed}), i.e., if $I$ satisfies the query and $I$ has a homomorphism to~$I'$ then~$I'$ also satisfies the query. Note that queries closed under homomorphisms are in particular \emph{monotone}, i.e., if $I$ satisfies the query and $I \subseteq I'$, then $I'$ also satisfies the query.

One well-known subclass of \ucqinf is \emph{bounded} \ucqinf: every bounded query in \ucqinf is logically equivalent to a \emph{union of conjunctive queries} (UCQ), without negation or inequalities. Recall that a \emph{conjunctive query} (CQ) is an existentially quantified conjunction of atoms, and a UCQ is a disjunction of CQs. For brevity, we omit existential quantification when writing UCQs, and abbreviate conjunction with a comma. The other \ucqinf queries are called \emph{unbounded}, and they can be seen as an infinite disjunction of CQs, with each disjunct corresponding to a model of the query.

A natural query language captured by \ucqinf is \emph{Datalog}, again without negation or inequalities. A Datalog program defines a signature of \emph{intensional predicates}, including a 0-ary predicate $\mathrm{Goal}()$, and consists of a set of \emph{rules} which explain how intensional facts can be \emph{derived} from other intensional facts and from the facts of the instance (called \emph{extensional}). The interpretation of the intensional predicates is defined by taking the (unique) least fixpoint of applying the rules, and the query holds if and only if the $\mathrm{Goal}()$ predicate can be derived. For formal definitions of this semantics, we refer the reader to the standard literature~\cite{abiteboul1995foundations}.
Datalog can in particular be used to express \emph{regular path queries} (RPQs)
and \emph{conjunctions of regular path queries with inverses}
(C2RPQs)~\cite{barcelo2013querying}.

As Datalog queries are homomorphism-closed, we can see each Datalog program as a \ucqinf, with the disjuncts intuitively corresponding to \emph{derivation trees} for the program.

\begin{exa}%
  \label{exa:datalog}
Consider the following Datalog program with one monadic intensional predicate $U$ over extensional signature $R, S, T$:
\begin{align*}
 R(x,y) &\rightarrow U(x), \\
 U(x), S(x, y) &\rightarrow U(y), \\
 U(x), T(x,y) &\rightarrow \mathrm{Goal()}.
\end{align*}
This program tests if the instance contains a path of
  facts $R(a_0, a_1), S(a_1, a_2), \ldots, S(a_{n-1}, a_n), \allowbreak T(a_n, a_{n+1})$ for
  some $n>0$, intuitively corresponding to the regular path query $R S^* T$.
  This is an unbounded \ucqinf.
\end{exa}

However, note that the class \ucqinf is a larger class than Datalog, because there are homomorphism-closed queries that are not expressible in Datalog~\cite{dawar2008datalog}.

\emph{Ontology-mediated queries}, or OMQs~\cite{BCLW14}, are another subclass of \ucqinf. An OMQ is a pair $(Q,\Tmc)$, where $Q$ is (typically) a UCQ, and $\Tmc$ is an ontology. A database instance $I$ \emph{satisfies} an OMQ $(Q,\Tmc)$ if the instance $I$ and the logical theory $\Tmc$ entail the query~$Q$ in the standard sense~-- see, e.g.,~\cite{BCLW14}, for details. There are ontological languages for OMQs based on \emph{description logics}~\cite{BCM07} and on \emph{existential rules}, also known as \emph{tuple-generating dependencies~(TGDs)}~\cite{CaGK-JAIR13,CaGL-JWS12}.
It is well known that every OMQ $(Q,\Tmc) \in (\text{UCQ},\text{TGD})$ is closed under homomorphisms. Thus, the dichotomy result of the paper applies to every OMQ from $(\text{UCQ},\text{TGD})$ over unary and binary predicates, which, in turn, covers several OMQ languages based on description logics.
There are also many OMQs that can be equivalently expressed as a query in Datalog or in disjunctive Datalog over an arity-two signature~\cite{BCLW14,EOS+-AAAI12,GoSc-KR12}, thus falling in the class \ucqinf. In particular, this is the case of any OMQ involving negation-free $\ALCHI$ (Theorem~6 of~\cite{BCLW14}), and of fragments of~$\ALCHI$, e.g., $\ELHI$, and $\ELI$ as in~\cite{JL12}.

\subsection*{Probabilistic query evaluation.}
We study the problem of probabilistic query evaluation over tuple-independent probabilistic databases. A \emph{tuple-independent probabilistic database~(TID)} over a signature~$\sigma$ is a pair $\pdb=(I, \pi)$ of a $\sigma$-instance $I$, and of a function~$\pi$ that maps every fact $F$ to a probability $\pi(F)$, given as a rational number in $[0, 1]$. Formally, a TID $\pdb=(I, \pi)$ defines the following probability distribution
over all \emph{possible worlds} $I' \subseteq I$:
\begin{align*}
\centering
\pi(I') \colonequals \left(\prod_{F \in I'} \pi(F)\right)\times\left(\prod_{F \in I' \setminus I} (1 - \pi(F))\right).
\end{align*}
Then, given a TID $\pdb=(I, \pi)$, the probability of a query~$Q$ relative to $\pdb$, denoted $\Pr_\pdb(Q)$, is given by the sum of the probabilities of the possible worlds that satisfy the query:
\[
\Pr_\pdb(Q) \coloneqq \sum_{{I' \subseteq I,I' \models Q}} \pi(I').
\]
The \emph{probabilistic query evaluation problem} (\PQE) for a query~$Q$, written $\PQE(Q)$, is then the task of computing $\Pr_\pdb(Q)$ given a TID $\pdb$ as input.

\subsection*{Complexity background.}
\FPTime is the class of functions $f: {\{0,1\}}^* \mapsto {\{0,1\}}^*$ computable by
a polynomial-time deterministic Turing machine. The class \sharpP, introduced by
Valiant in~\cite{Valiant79}, contains the
computation problems that can be expressed as the number of accepting paths of a nondeterministic polynomial-time Turing machine. Equivalently, a function ${f: {\{0,1\}}^* \mapsto \Nbb}$ is in \sharpP if there
exists a polynomial $p: \Nbb \mapsto \Nbb$  and a polynomial-time deterministic Turing machine $M$ such that for every $x \in {\{0,1\}}^*$, it holds that:
\[
  f(x)=\card{\left\{ y \in {\{0,1\}}^{p(\card{x})} \mid M~\emph{answers 1 on the input $(x, y)$}~\right\}}.
\]

For a query~$Q$, we study the \emph{data complexity} of $\PQE(Q)$, which is measured as a function of the input instance~$I$, i.e., the signature and~$Q$ are fixed.
For a large class of queries, in particular for any UCQ~$Q$, the problem $\PQE(Q)$ is in the complexity class $\mathrm{FP}^{\#P}$: we can use a nondeterministic Turing machine to guess a possible world according to the probability distribution of the TID (i.e., each possible world is obtained in a number of runs proportional to its probability), and then check in polynomial time data complexity if~$Q$ holds, with polynomial-time postprocessing to renormalize the number of runs to a probability.
Our goal in this work is to show that the problem is also \#P-hard.

To show \#P-hardness, we use \emph{polynomial-time Turing reductions}~\cite{Cook71}. A function $f$ is \sharpP-complete under polynomial time Turing reductions if it is in $\sharpP$ and every $g \in \sharpP$ is in $\FPTime^f$. Polynomial-time Turing reductions are the most common reductions for the class \sharpP and they are the reductions used to show \sharpP-hardness in the dichotomy of Dalvi and Suciu~\cite{dalvi2012dichotomy}, so we use them throughout this work.

\subsection*{Problems.}
We will show hardness by reducing from two well-known \sharpP-hard problems. For some queries, we reduce from \pptwodnf~\cite{provan1983complexity}, which is a standard tool to show hardness of unsafe UCQs. The original problem uses Boolean formulae; here, we give an equivalent rephrasing in terms of bipartite graphs.
\begin{defi}%
	\label{def:pp2dnf}
	Given a bipartite graph $H = (A, B, C)$ with edges $C \subseteq A \times B$, a \emph{possible world} of~$H$ is a pair $\omega = (A', B')$ with $A' \subseteq A$ and $B' \subseteq B$. We call the possible world \emph{good} if it is not an independent set, i.e., if one vertex of~$A'$ and one vertex of~$B'$ are adjacent in~$C$; and call it \emph{bad} otherwise.
	The \emph{positive partitioned 2DNF problem ({\normalfont\pptwodnf})} is the following: given a bipartite graph, compute how many of its possible worlds are good.
\end{defi}
It will be technically convenient to assume that~$H$ is connected. This is clearly without loss of generality, as otherwise the number of good possible worlds is simply obtained as the product of the number of good possible worlds of each connected component of~$H$.

For other queries, we reduce from a different problem, known as the \emph{undirected st-connectivity problem} (\stcon)~\cite{provan1983complexity}:
\begin{defi}%
  \label{def:stcon}
  An \emph{st-graph} is an undirected graph $G = (W,
  C)$ with two distinguished vertices $s \in W$ and $t \in W$.
  A \emph{possible world} of~$G$ is a subgraph $\omega = (W, C')$ with~$C' \subseteq C$.
  We call the possible world \emph{good} if $C'$ contains a path connecting~$s$
  and~$t$, and \emph{bad} otherwise.
  The \emph{source-to-target undirected reachability problem ({\normalfont\stcon})} is the
  following: given an st-graph, compute how many of its possible worlds are
  good.
\end{defi}

\section{Result Statement}%
\label{sec:result}

The goal of this paper is to extend the dichotomy of Dalvi and Suciu~\cite{dalvi2012dichotomy} on \PQE\ for UCQs. Their result states:
\begin{thmC}[\cite{dalvi2012dichotomy}]%
  \label{thm:dalvisuciu}
  Let $Q$ be a UCQ\@. Then, $\PQE(Q)$ is either in {\normalfont\FPTime} or it is {\normalfont\sharpP}-hard.
\end{thmC}

Following Dalvi and Suciu's terminology, we call a UCQ \emph{safe} if $\PQE(Q)$ is in \FPTime, and \emph{unsafe} otherwise. This dichotomy characterizes the complexity of PQE for UCQs, but does not apply to other homomorphism-closed queries beyond UCQs. Our contribution, when restricting to the arity-two setting, is to generalize this dichotomy to \ucqinf, i.e., to \emph{any} query closed under homomorphisms. Specifically, we show that all such queries are intractable unless they are equivalent to a safe UCQ\@.
\begin{thm}[Dichotomy]%
  \label{thm:main}
  Let $Q$ be a \ucqinf over an arity-two signature. Then, either~$Q$ is equivalent to a safe UCQ and $\PQE(Q)$ is in {\normalfont\FPTime}, or it is not and $\PQE(Q)$ is {\normalfont\sharpP}-hard.
\end{thm}
Our result relies on the dichotomy of Dalvi and Suciu for \ucqinf queries that are equivalent to UCQs. The key point is then to show intractability for \emph{unbounded} \ucqinf queries. Hence, our technical contribution is to show:
\begin{thm}%
  \label{thm:main2}
  Let $Q$ be an unbounded \ucqinf query over an arity-two signature. Then, $\PQE(Q)$ is {\normalfont\sharpP}-hard.
\end{thm}

This result applies to the very general class of unbounded \ucqinf. It implies
in particular that the PQE problem is \#P-hard for all Datalog queries that are
not equivalent to a UCQ, as in Example~\ref{exa:datalog}: this is the case of
all Datalog queries except the ones that are
nonrecursive or where recursion is \emph{bounded}~\cite{Hillebrand95}.

\subsection*{Effectiveness and uniformity.}
We do not study whether our dichotomy result in Theorem~\ref{thm:main} is effective, i.e., we do not study the problem of determining, given a query, whether it is safe or unsafe. The dichotomy of Theorem~\ref{thm:dalvisuciu} for UCQs is effective via the algorithm of~\cite{dalvi2012dichotomy}: this algorithm has a super-exponential bound (in the query), with the precise complexity being open.
Our dichotomy concerns the very general query language \ucqinf, and its effectiveness depends on how the input is represented: to discuss this question, we need to restrict queries to some syntactically defined fragment. If we restrict to Datalog queries, it is not clear whether our dichotomy is effective, because it is undecidable, given an arbitrary Datalog program as input, to determine whether it is bounded~\cite{Gaifman93}. This means that there is little hope for our dichotomy to be decidable over arbitrary Datalog queries, but on its own it does not imply undecidability, so the question remains open. However, our dichotomy is effective for query languages for which boundedness is decidable, e.g., monadic Datalog, its generalization GN-Datalog~\cite{benedikt2015complexity}, C2RPQs~\cite{barcelo2019boundedness}, or ontology-mediated query answering with guarded existential rules~\cite{BBLP-18}.

For unsafe queries, we also do not study the complexity of reduction \emph{as a function of the query}, or whether this problem is even decidable. All that matters is that, once the query is fixed, some reduction procedure exists, which can be performed in polynomial time \emph{in the input instance}. Such uniformity problems seem unavoidable, given that our language \ucqinf is very general and includes some queries for which non-probabilistic evaluation is not even decidable, e.g., ``there is a path from $R$ to $T$ whose length is the index of a Turing machine that halts''. We leave for future work the study of the query complexity of our reduction when restricting to better-behaved query languages such as Datalog or RPQs.

\subsection*{Proof outline.}
Theorem~\ref{thm:main2} is proven in Sections~\ref{sec:pp2dnf}--\ref{sec:ustcon}.
There are two cases, depending on the query. We study the first case in Section~\ref{sec:pp2dnf}, which covers queries for which we can find a model with a so-called \emph{non-iterable edge}. Intuitively, this is a model where we can make the query false by replacing the edge by a back-and-forth path of some length between two neighboring facts that it connects. For such queries, we can show hardness by a reduction from \pptwodnf, essentially like the hardness proof for the query $Q_0 : R(w, x), S(x, y), T(y, z)$ which is the arity-two variant of the unsafe query of~\cite[Theorem~5.1]{DaSu07}. This hardness proof covers some bounded queries (including $Q_0$) and some unbounded ones.

In Section~\ref{sec:findhard}, we present a new ingredient, to be used in the second case, i.e., when there is no model with a non-iterable edge. We show that any unbounded query must always have a model with an edge that is \emph{tight}, i.e., the query no longer holds if we replace that edge with two copies, one copy connected only to the first element and another copy connected only to the second element.
What is more, we can find a model with a tight edge which is \emph{minimal} in some sense, which we call a \emph{minimal tight pattern}.

In Section~\ref{sec:ustcon}, we use minimal tight patterns for the second case, covering unbounded queries that have a minimal tight pattern where the tight edge of the pattern is iterable. This applies to all queries to which Section~\ref{sec:pp2dnf} did not apply (and also to some queries to which it did).
Here, we reduce from the \stcon problem: intuitively, we use the iterable edge for a kind of reachability test, and we use the minimality and tightness of the pattern to show the soundness and completeness of the reduction.

\subsection*{Generalizations.}
In Section~\ref{sec:generalizations}, we give two generalizations of our result. First, we observe that our reductions only use tuple probabilities from $\{0,0.5,1\}$. This means that all $\sharpP$-hardness results hold even when restricting the probabilistic query evaluation problem to the so-called \emph{generalized model counting problem} studied for instance in~\cite{KenigSuciu20}, so we can also state our dichotomy in this context. Second, we show that all our results also apply when we consider signatures featuring predicates with arity 1 and 2.

\section{Hardness with Non-Iterable Edges}%
\label{sec:pp2dnf}

\newcommand{\drawnewkey}{
    \tcbox[colframe=black!80,colback=white, boxsep=-2mm, boxrule=0.2mm, ]{
	\begin{tikzpicture}[yscale=0.9, xscale=1]
	\draw  (0, 2.75) edge[-, cyan, very thick]  (.6, 2.75);
	\node[align=left,text width=3.3cm] (cyan) at (2.3, 2.75) {$e$};
	\draw  (2, 3) edge[-, thick, blue]  (2.6, 3);
	\node[align=left,text width=3.3cm] (blue) at (4.3, 3) {$e_{\ll}$ with $F_{\ll}$};
	\draw  (2, 2.5) edge[-, thick, blue,dashed]  (2.6, 2.5);
	\node[align=left,text width=3.3cm] (blued) at (4.3, 2.5) {$e_{\ll}$ without $F_{\ll}$};
	\draw  (6, 3) edge[-, thick, black]  (6.6, 3);
	\node[align=left,text width=3.3cm] (black) at (8.3, 3) {$e_{\rr}$ with $F_{\rr}$};
	\draw  (6, 2.5) edge[-, thick, black,dashed]  (6.6, 2.5);
	\node[align=left,text width=3.3cm] (blackd) at (8.3, 2.5) {$e_{\rr}$ without $F_{\rr}$};
	\draw  (10, 2.75) edge[-, orange]  (10.6, 2.75);
	\node[align=left,text width=2.7cm] (orange) at (12.3, 2.75) {other incident edges};
	\end{tikzpicture}
    }
}
\newcommand{\drawnewkeyb}{
    \tcbox[colframe=black!80,colback=white, boxsep=-0.2mm, boxrule=0.2mm, ]{
	\begin{tikzpicture}[yscale=0.9, xscale=1]
	\draw  (0, 3) edge[-, cyan, very thick]  (.6, 3);
	\node[align=left,text width=3.3cm] (cyan) at (2.3, 3) {$e$ with $F_{\mm}$};
	\draw  (0, 2.5) edge[-, cyan, very thick, dashed]  (.6, 2.5);
	\node[align=left,text width=3.3cm] (cyan) at (2.3, 2.5) {$e$ without $F_{\mm}$};
	\draw  (3.5, 3) edge[-, thick, blue]  (4.1, 3);
	\node[align=left,text width=3.3cm] (blue) at (5.8, 3) {$e_{\ll}$ with $F_{\ll}$};
	\draw  (3.5, 2.5) edge[-, thick, blue,dashed]  (4.1, 2.5);
	\node[align=left,text width=3.3cm] (blued) at (5.8, 2.5) {$e_{\ll}$ without $F_{\ll}$};
	\draw  (7, 3) edge[-, thick, black]  (7.6, 3);
	\node[align=left,text width=3.3cm] (black) at (9.3, 3) {$e_{\rr}$ with $F_{\rr}$};
	\draw  (7, 2.5) edge[-, thick, black,dashed]  (7.6, 2.5);
	\node[align=left,text width=3.3cm] (blackd) at (9.3, 2.5) {$e_{\rr}$ without $F_{\rr}$};
	\draw  (10.5, 2.75) edge[-, orange]  (11.1, 2.75);
	\node[align=left,text width=2.5cm] (orange) at (12.5, 2.75) {other incident edges};
	\end{tikzpicture}
    }
}

In this section, we present the hardness proof for the first case where we can find a model of the query with a \emph{non-iterable edge}. This notion will be defined relative to an \emph{incident pair} of a \emph{non-leaf edge}:

\begin{defi}
Let $I$ be an instance. We say that an element $u \in \dom(I)$ of~$I$ is a \emph{leaf} if it occurs in
only one undirected edge. We say that an edge (directed or undirected) is a \emph{leaf edge} if one of its elements (possibly both) is a leaf; otherwise, it is a \emph{non-leaf edge}.

Let $I$ be an instance and let $e = (u, v)$ be a non-leaf edge of $I$. A $\sigma^\leftrightarrow$-fact of~$I$ is \emph{left-incident} to~$e$ if it is of the form $R_{\ll}(l, u)$ with $l \notin \{u, v\}$. It is \emph{right-incident} to~$e$ if it is of the form $R_{\rr}(v, r)$ with $r \notin \{u, v\}$. An \emph{incident pair} of~$e$ is a pair of $\sigma^\leftrightarrow$-facts $\Pi= (F_{\ll}, F_{\rr})$, where $F_{\ll}$ is left-incident to~$e$ and $F_{\rr}$ is right-incident to~$e$.
We write $I_{e,\Pi}$ to denote an instance $I$ with a distinguished non-leaf edge $e$ and a distinguished incident pair $\Pi$ of~$e$ in~$I$.
\end{defi}
Note that an incident pair chooses two incident \emph{facts} (not edges): this is intuitively because in the PQE problem, we will give probabilities to single facts and not edges. It is clear that every non-leaf edge~$e$ must have an incident pair, as we can pick $F_{\ll}$ and $F_{\rr}$ from the edges incident to~$u$ and~$v$ which are not~$e$.
Moreover, we must have $F_{\ll} \neq F_{\rr}$, and neither $F_{\ll}$ nor~$F_{\rr}$ can be unary facts. However, as the relations $R_{\ll}$ and $R_{\rr}$ are $\sigma^\leftrightarrow$-relations, we may have $R_{\ll} = R_{\rr}$ or $R_{\ll} = R_{\rr}^-$, and the elements $l$ and $r$ may be equal if the edge $(u, v)$ is part of a triangle with some edges $\{u, w\}$ and $\{v, w\}$.

Let us illustrate the notion of incident pair on an example.
\begin{exa}
 Given an instance $I = \{ R(a, b), T(b, b), S(c, b), R(d, c) \}$, the edge $(b, c)$ is non-leaf and the only possible incident pair for it  is $(R(a, b), R^-(c, d))$.
\end{exa}

We can now define the \emph{iteration process} on an instance~$I_{e,\Pi}$, which intuitively replaces the edge~$e$ by a path of copies of~$e$, keeping the facts of~$\Pi$ at the beginning and end of the path, and copying all other incident facts.
Note that, while the instances that we work with are over the signature~$\sigma$, we will see them as $\sigma^\leftrightarrow$ instances in the definition of this process, e.g., when creating copies of facts, in particular of the $\sigma^\leftrightarrow$-facts $F_{\ll}$ and $F_{\rr}$; but the facts that we actually create are $\sigma$-facts, and the resulting instance is a $\sigma$-instance.
The iteration process is represented in Figure~\ref{fig:disso}, and defined formally below:
\begin{defi}%
	\label{def:iteration}
	Let $I_{e,\Pi}$ be a $\sigma$-instance where $e = (u, v)$, $\Pi = (F_{\ll}, F_{\rr})$, $F_{\ll} = R_{\ll}(l, u)$, $F_{\rr} = R_{\rr}(v, r)$, and  let $n \geq 1$.
    The \emph{$n$-th iterate} of~$e$ in~$I$ relative to~$\Pi$, denoted $I_{e,\Pi}^n$, is a $\sigma$-instance with domain $\dom(I_{e,\Pi}^n) \colonequals \dom(I) \cup \{u_2, \ldots, u_n\} \cup \{v_1, \ldots, v_{n-1}\}$, where the new elements are fresh, and where we use  $u_1$ to refer to~$u$ and $v_n$ to refer to~$v$ for convenience.
    The facts of $I_{e,\Pi}^n$ are defined by applying the following steps:
	\begin{itemize}
		\item \emph{Copy non-incident facts:} Initialize $I_{e,\Pi}^n$ as the induced subinstance of~$I$ on $\dom(I) \setminus \{u, v\}$.
                \item \emph{Copy incident facts $F_{\ll}$ and $F_{\rr}$:}
                  Add $F_{\ll}$ and $F_{\rr}$ to~$I_{e,\Pi}^n$, using $u_1$ and $v_n$, respectively.
                \item \emph{Copy other left-incident facts:}
                  For each $\sigma^\leftrightarrow$-fact $F_{\ll}' = R_{\ll}'(l', u)$ of~$I$ that is left-incident to~$e$ (i.e., $l' \notin \{u, v\}$) and where $F_{\ll}' \neq F_{\ll}$, add to~$I_{e,\Pi}^n$ the fact $R_{\ll}'(l', u_i)$ for each $1 \leq i \leq n$.
                \item \emph{Copy other right-incident facts:} For each $\sigma^\leftrightarrow$-fact $F_{\rr}' = R_{\rr}'(v, r')$ of~$I$ that is right-incident to~$e$ (i.e., $r' \notin \{u, v\}$) and where $F_{\rr}' \neq F_{\rr}$, add to~$I_{e,\Pi}^n$ the fact $R_{\rr}'(v_i, r')$ for each $1 \leq i \leq n$.
                \item \emph{Create copies of~$e$:} Copy the edge $e$ (in the sense defined in the Preliminaries) on the following pairs: $(u_i, v_i)$ for $1 \leq i \leq n$, and $(u_{i+1}, v_i)$ for $1 \leq i \leq n-1$.
	\end{itemize}
\end{defi}

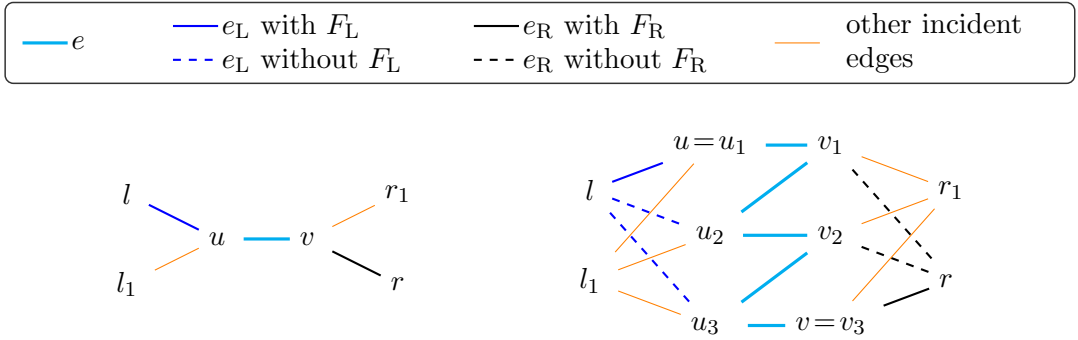
\begin{figure}
  \drawnewkey%
	\null\hfill\begin{tikzpicture}[node distance = 0.5cm,shorten <=3pt,-,yscale=.6,xscale=.6]
	\node (dumb) at (0, -3.1) {};
	\node (vl) at (0, 0) {$l$};
	\node (vll) at (0, -2) {$l_1$};
	\node (u) at (2, -1) {$u$};
	\node (v) at (4, -1) {$v$};
	\node (ur) at (6, -2) {$r$};
	\node (urr) at (6, 0) {$r_1$};
	\draw (u) edge[-,cyan,very thick] (v);
	\draw (vl) edge[-,thick,blue] (u);
	\draw (vll) edge[-,orange] (u);
	\draw (v) edge[thick,-,black] (ur);
	\draw (v) edge[-,orange] (urr);
	\end{tikzpicture}
	\hfill
	\begin{tikzpicture}[node distance = 0.5cm,shorten <=3pt,-,yscale=.6,xscale=.8]
	\node (vl) at (0, -2) {$l$};
	\node (vll) at (0, -4) {$l_1$};
	\node (u1) at (2, -1) {$u\!=\!u_1$};
	\node (u2) at (2, -3) {$u_2$};
	\node (u3) at (2, -5) {$u_3~$};
	\node (v1) at (4, -1) {$v_1$};
	\node (v2) at (4, -3) {$v_2$};
	\node (v3) at (4, -5) {$v\!=\!v_3$};
	\node (ur) at (6, -4) {$r~$};
	\node (urr) at (6, -2) {$r_1$};
	\draw (u1) edge[-,cyan,very thick] (v1);
	\draw (u2) edge[-,cyan,very thick] (v1);
	\draw (u2) edge[-,cyan,very thick] (v2);
	\draw (u3) edge[-,cyan,very thick] (v2);
	\draw (u3) edge[-,cyan,very thick] (v3);
	\draw (vl) edge[-,thick,blue] (u1);
	\draw (vl) edge[-,dashed,thick,blue] (u2);
	\draw (vl) edge[-,dashed,thick,blue] (u3);
	\draw (vll) edge[-,orange] (u1);
	\draw (vll) edge[-,orange] (u2);
	\draw (vll) edge[-,orange] (u3);
	\draw (v3) edge[thick,-,black] (ur);
	\draw (v1) edge[-,dashed,thick,black] (ur);
	\draw (v2) edge[-,dashed,thick,black] (ur);
	\draw (v1) edge[-,orange] (urr);
	\draw (v2) edge[-,orange] (urr);
	\draw (v3) edge[-,orange] (urr);
	\end{tikzpicture}\hfill\null
        \caption{Example of iteration from an instance $I_{e,\Pi}$ (left) to~$I_{e,\Pi}^3$ (middle). We write $\Pi = (F_{\ll}, F_{\rr})$ and call $e_{\ll}$ and $e_{\rr}$ the edges of~$F_{\ll}$ and~$F_{\rr}$.
        Each line represents an edge covering in general multiple $\sigma^\leftrightarrow$-facts.
        A key is given at the top.
            }%
	\label{fig:disso}
\end{figure}

\noindent
Note that, for $n=1$, we obtain exactly the original instance. Intuitively, we replace~$e$ by a path going back-and-forth between copies of~$u$ and~$v$ (and traversing~$e$ alternatively in one direction and another). The intermediate vertices have the same incident facts as the original endpoints except that we have not copied the left-incident fact and the right-incident fact of the incident pair.

We first notice that larger iterates have homomorphisms back to smaller iterates:

\begin{obs}%
	\label{obs:iterhomom}
	For any instance~$I$, for any non-leaf edge $e$ of~$I$, for any incident pair~$\Pi$ for~$e$, and for any
	$1\leq i \leq j$, it holds that $I_{e,\Pi}^j$ has a homomorphism to $I_{e,\Pi}^i$.
\end{obs}

\begin{proof}
  Simply merge $u_i, \ldots, u_j$, and merge $v_i, \ldots, v_j$.
\end{proof}

Hence, choosing an instance~$I$ that satisfies $Q$, a non-leaf edge~$e$ of~$I$, and an incident pair~$\Pi$, there are two possibilities. Either all iterates $I_{e,\Pi}^n$ satisfy~$Q$, or there is some iterate $I_{e,\Pi}^{n_0}$ with $n_0 > 1$ that violates $Q$ (and, by Observation~\ref{obs:iterhomom}, all subsequent iterates also do). We call~$e$ \emph{iterable} relative to~$\Pi$ in the first case, and \emph{non-iterable} in the second case:

\begin{defi}%
\label{def:iterable}
  A non-leaf edge~$e$ of a model $I$ of a query~$Q$ is \emph{iterable relative to an incident pair~$\Pi$} if $I_{e,\Pi}^n$ satisfies~$Q$ for each $n \geq 1$; otherwise, it is \emph{non-iterable relative to~$\Pi$}. We call $e$ \emph{iterable} if it is iterable relative to some incident pair, and \emph{non-iterable} otherwise.
\end{defi}

The goal of this section is to show that if a query~$Q$ has a model with a non-leaf edge which is not iterable, then $\PQE(Q)$ is intractable:

\begin{thm}%
	\label{thm:pp2dnfred}
	For every \ucqinf $Q$, if $Q$ has a model $I$ with a non-leaf edge~$e$ that is non-iterable, then $\PQE(Q)$ is \#P-hard.
\end{thm}

Let us illustrate on an example how to apply this result:
\begin{exa}%
  \label{exa:rstrpq}
Consider the RPQ $R S^* T$. This query has a model $\{R(a, b), S(b, c), T(c, d)\}$ with an edge $(b, c)$ that is non-leaf and non-iterable. Indeed its iterate with $n = 2$ relative to the only possible incident pair yields
$\{R(a, b), S(b, c'), S(b', c'), \allowbreak S(b', c), T (c, d)\}$ which does not satisfy the query.
Hence, Theorem~\ref{thm:pp2dnfred} shows that PQE is \#P-hard for this RPQ\@.
Importantly, the choice of the model matters, as this query also has models where all non-leaf edges are iterable, for instance $\{R(a, b), S(b, c), T(c, d), R(a', b'), S(b', c'), \allowbreak T(c', d')\}$, or $\{R(a, b), T(b, c)\}$ which has no non-leaf edge at all.
\end{exa}

Note that Theorem~\ref{thm:pp2dnfred} does not assume that the query is unbounded, and also applies to some bounded queries.
For instance, the unsafe CQ $Q_0: R(w, x), S(x,y), T(y, z)$ can be shown to be unsafe using this result, with the model $\{R(a, b), S(b, c), T(c, d)\}$ and edge $(b, c)$.
However, Theorem~\ref{thm:pp2dnfred} is too coarse to show \#P-hardness for all unsafe UCQs; for instance,  it does not cover ${Q_0': R(x, x), S(x, y), T(y, y)}$, or ${Q_1: (R(w, x), S(x, y))
\lor (S(x, y), T(y, z))}$. It will nevertheless be sufficient for our purposes when studying \emph{unbounded queries}, as we will see in the next sections.

Hence, in the rest of this section, we prove Theorem~\ref{thm:pp2dnfred}. Let $I_{e,\Pi}$ be the instance with the non-iterable, non-leaf edge, and let us take the smallest $n_0 > 1$ such that $I_{e,\Pi}^{n_0}$ violates the query. The idea is to use $I_{e,\Pi}^{n_0}$ to show hardness of PQE by reducing from
\pptwodnf (Definition~\ref{def:pp2dnf}). Thus, let us explain how we can use $I_{e,\Pi}$ to code a bipartite graph~$H$ in polynomial time into a TID $\pdb$.
The definition of this coding does not depend on the query~$Q$, but we will use the properties of
$I_{e,\Pi}$ and~$n_0$ to argue that it defines a reduction between \pptwodnf and PQE, i.e., there is a
correspondence between the possible worlds of~$H$ and the possible worlds of~$\pdb$, such that good possible worlds of~$H$ are mapped to possible worlds of~$\pdb$ which satisfy~$Q$. Let us first define the coding, which we also illustrate on an example in Figure~\ref{fig:pp2dnfcode}:

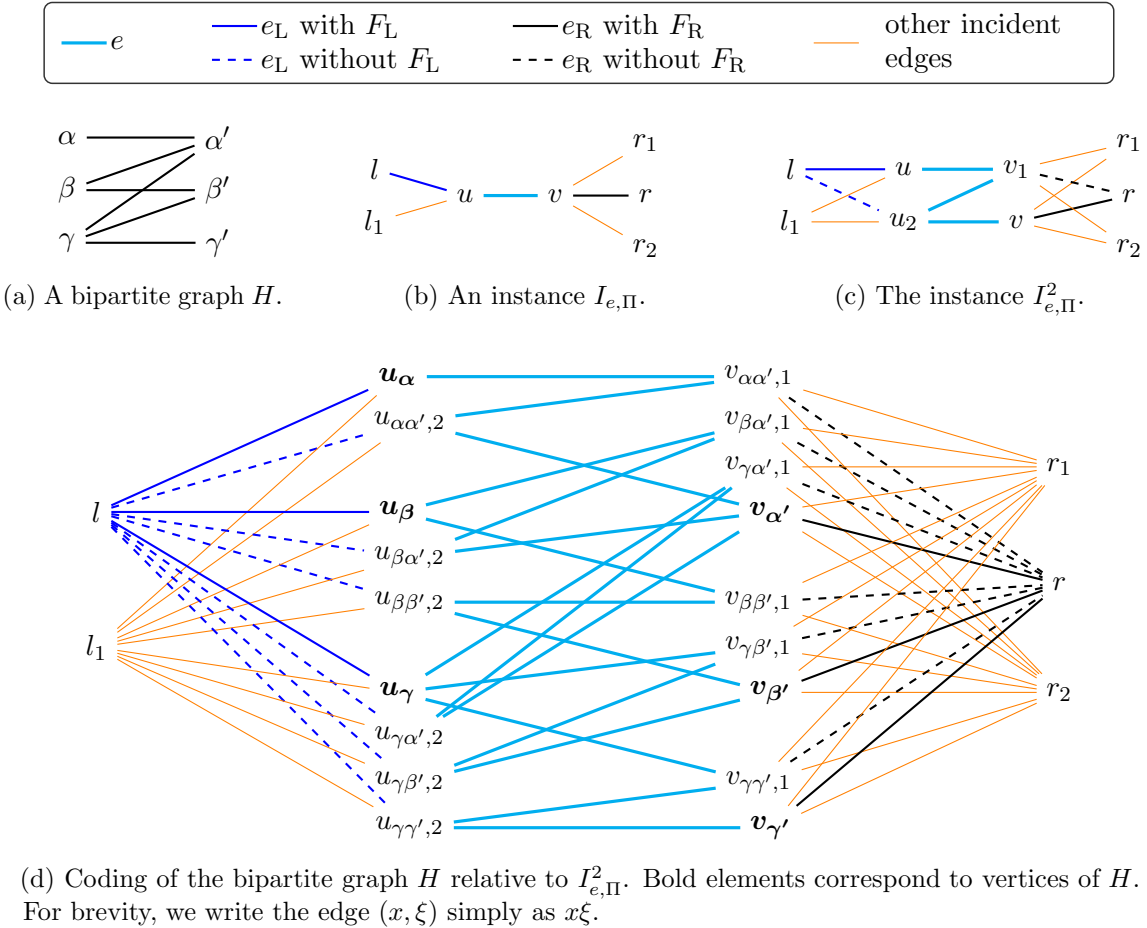
\begin{figure}[t]
	\drawnewkey%
	\bigskip

	\begin{subfigure}[t]{1.45in}
		\centering
		\begin{tikzpicture}[yscale=.7]
		\node (a) at (0, 0) {$\alpha$};
		\node (b) at (0, -1) {$\beta$};
		\node (c) at (0, -2) {$\gamma$};
		\node (alpha) at (2, 0) {$\alpha'$};
		\node (beta) at (2, -1) {$\beta'$};
		\node (gamma) at (2, -2) {$\gamma'$};
		\draw (a) edge[-, thick] (alpha);
		\draw (b) edge[-, thick] (alpha);
		\draw (b) edge[-, thick] (beta);
		\draw (c) edge[-, thick] (beta);
		\draw (c) edge[-, thick] (gamma);
		\draw (c) edge[-, thick] (alpha);
		\end{tikzpicture}
		\caption{A bipartite graph $H$.}\label{fig:pp2dnfcode-a}
	\end{subfigure}
	\hfill
	\begin{subfigure}[t]{1.8in}
		\begin{tikzpicture}[yscale=.7]
		\node (v1) at (0, -0.5) {$l$};
		\node (v2) at (0, -1.5) {$l_1$};
		\node (u) at (1.2, -1) {$u$};
		\node (v) at (2.4, -1) {$v$};
		\node (u1) at (3.6, 0) {$r_1$};
		\node (u2) at (3.6, -1) {$r$};
		\node (u3) at (3.6, -2) {$r_2$};
		\draw (u) edge[-, thick, cyan, very thick] (v);
		\draw (v1) edge[-,thick, blue]  (u);
		\draw (v2) edge[-, orange] (u);
		\draw (v) edge[-, orange] (u1);
		\draw (v) edge[-,thick, black] (u2);
		\draw (v) edge[-,orange] (u3);
		\end{tikzpicture}
		\caption{An instance $I_{e,\Pi}$.}\label{fig:pp2dnfcode-b}
	\end{subfigure}
	\hfill
	\begin{subfigure}[t]{2in}
		\centering
		\begin{tikzpicture}[yscale=.7]
		\node (v1) at (0, -0.5) {${l}$};
		\node (v2) at (0, -1.5) {${l_1}$};
		\node (u) at (1.5, -0.5) {$u$};
		\node (ua) at (1.5, -1.5) {$u_2$};
		\node (va) at (3, -0.5) {$v_1$};
		\node (v) at (3, -1.5) {$v$};
		\node (u1) at (4.5, 0) {${r_1}$};
		\node (u2) at (4.5, -1) {${r}$};
		\node (u3) at (4.5, -2) {${r_2}$};
		\draw (v1) edge[-,thick, blue] (u);
		\draw (v1) edge[-,dashed,thick, blue] (ua);
		\draw (v2) edge[-, orange] (u);
		\draw (v2) edge[-,orange] (ua);
		\draw (u) edge[-,thick, cyan,very thick]  (va);
		\draw (ua) edge[-,thick, cyan,very thick](v);
		\draw (ua) edge[-,thick, cyan,very thick](va);
		\draw (va) edge[-,orange] (u1);
		\draw (v) edge[-,orange] (u1);
		\draw (va) edge[-,thick, dashed,black] (u2);
		\draw (v) edge[-,thick, black] (u2);
		\draw (va) edge[-, orange] (u3);
		\draw (v) edge[-, orange] (u3);
		\end{tikzpicture}
		\caption{The instance $I_{e,\Pi}^2$.}\label{fig:pp2dnfcode-c}
	\end{subfigure}

	\bigskip
	\medskip

	\null\hfill\begin{subfigure}[t]{.97\linewidth}
		\centering
		\begin{tikzpicture}[xscale=1.6,yscale=1.2]
		\node (v1) at (0, -0.5) {${l}$};
		\node (v2) at (0, -2) {${l_1}$};
		\node (ua) at (2.5, 1) {$\bm{u_\alpha}$};
		\node (uaalpha) at (2.6, 0.5) {$u_{\alpha\alpha',2}$};
		\node (ub) at (2.5, -0.5) {$\bm{u_\beta}$};
		\node (ubalpha) at (2.6, -1) {$u_{\beta\alpha',2}$};
		\node (ubbeta) at (2.6, -1.5) {$u_{\beta\beta',2}$};
		\node (uc) at (2.5, -2.5) {$\bm{u_\gamma}$};
		\node (ucalpha) at (2.6, -3) {$u_{\gamma\alpha',2}$};
		\node (ucbeta) at (2.6, -3.5) {$u_{\gamma\beta',2}$};
		\node (ucgamma) at (2.6, -4) {$u_{\gamma\gamma',2}$};
		\node (vaalpha) at (5.5, 1) {$v_{\alpha\alpha',1}$};
		\node (vbalpha) at (5.5, 0.5) {$v_{\beta\alpha',1}$};
		\node (vcalpha) at (5.5, 0) {$v_{\gamma\alpha',1}$};
		\node (valpha) at (5.6, -0.5) {$\bm{v_{\alpha'}}$};
		\node (vbbeta) at (5.5, -1.5) {$v_{\beta\beta',1}$};
		\node (vcbeta) at (5.5, -2) {$v_{\gamma\beta',1}$};
		\node (vbeta) at (5.6, -2.5) {$\bm{v_{\beta'}}$};
		\node (vcgamma) at (5.5, -3.5) {$v_{\gamma\gamma',1}$};
		\node (vgamma) at (5.6, -4) {$\bm{v_{\gamma'}}$};
		\node (u1) at (8, 0) {${r_1}$};
		\node (u2) at (8, -1.3) {${r}$};
		\node (u3) at (8, -2.5) {${r_2}$};
		\draw (v1) edge[-, thick, blue] (ua);
		\draw (v1) edge[-, thick, dashed,blue] (uaalpha);
		\draw (v2) edge[-, orange] (ua);
		\draw (v2) edge[-, orange] (uaalpha);
		\draw (v1) edge[-, thick, blue] (ub);
		\draw (v1) edge[-, thick, dashed,blue] (ubalpha);
		\draw (v1) edge[-, thick, dashed,blue] (ubbeta);
		\draw (v2) edge[-, orange] (ub);
		\draw (v2) edge[-, orange] (ubalpha);
		\draw (v2) edge[-, orange] (ubbeta);
		\draw (v1) edge[-, thick, blue] (uc);
		\draw (v1) edge[-, thick, dashed,blue] (ucalpha);
		\draw (v1) edge[-, thick, dashed,blue] (ucbeta);
		\draw (v1) edge[-, thick, dashed,blue] (ucgamma);
		\draw (v2) edge[-, orange] (uc);
		\draw (v2) edge[-, orange] (ucalpha);
		\draw (v2) edge[-, orange] (ucbeta);
		\draw (v2) edge[-, orange] (ucgamma);
		\draw (ua) edge[-,  cyan,very thick] (vaalpha);
		\draw (uaalpha) edge[-, cyan,very thick] (vaalpha);
		\draw (uaalpha) edge[-, cyan,very thick] (valpha);
		\draw (ub) edge[-,  cyan,very thick] (vbalpha);
		\draw (ubalpha) edge[-, cyan,very thick] (vbalpha);
		\draw (ubalpha) edge[-, cyan,very thick] (valpha);
		\draw (ub) edge[-,  cyan,very thick] (vbbeta);
		\draw (ubbeta) edge[-, cyan,very thick] (vbbeta);
		\draw (ubbeta) edge[-, cyan,very thick] (vbeta);
		\draw (uc) edge[-,  cyan,very thick] (vcbeta);
		\draw (ucbeta) edge[-, cyan,very thick] (vcbeta);
		\draw (ucbeta) edge[-, cyan,very thick] (vbeta);
		\draw (uc) edge[-,  cyan,very thick] (vcgamma);
		\draw (ucgamma) edge[-, cyan,very thick] (vcgamma);
		\draw (ucgamma) edge[-, cyan,very thick] (vgamma);
		\draw (uc) edge[-, cyan,very thick] (vcalpha);
		\draw (ucalpha) edge[-, cyan,very thick] (vcalpha);
		\draw (ucalpha) edge[-, cyan,very thick] (valpha);
		\draw (vaalpha) edge[-, orange] (u1);
		\draw (vaalpha) edge[-, thick,dashed, black] (u2);
		\draw (vaalpha) edge[-, orange] (u3);
		\draw (vbalpha) edge[-, orange] (u1);
		\draw (vbalpha) edge[-, thick,dashed, black] (u2);
		\draw (vbalpha) edge[-, orange] (u3);
		\draw (vcalpha) edge[-, orange] (u1);
		\draw (vcalpha) edge[-, thick,dashed, black] (u2);
		\draw (vcalpha) edge[-, orange] (u3);
		\draw (valpha) edge[-, orange] (u1);
		\draw (valpha) edge[-, thick,black] (u2);
		\draw (valpha) edge[-, orange] (u3);
		\draw (vbbeta) edge[-, orange] (u1);
		\draw (vbbeta) edge[-, thick,dashed, black] (u2);
		\draw (vbbeta) edge[-, orange] (u3);
		\draw (vcbeta) edge[-, orange] (u1);
		\draw (vcbeta) edge[-, thick,dashed, black] (u2);
		\draw (vcbeta) edge[-, orange] (u3);
		\draw (vbeta) edge[-, orange] (u1);
		\draw (vbeta) edge[-, thick,black] (u2);
		\draw (vbeta) edge[-, orange] (u3);
		\draw (vcgamma) edge[-, orange] (u1);
		\draw (vcgamma) edge[-, thick,dashed, black] (u2);
		\draw (vcgamma) edge[-, orange] (u3);
		\draw (vgamma) edge[-, orange] (u1);
		\draw (vgamma) edge[-, thick,black] (u2);
		\draw (vgamma) edge[-, orange] (u3);
		\end{tikzpicture}
		\caption{Coding of the bipartite graph $H$ relative to $I_{e,\Pi}^2$. Bold elements correspond to vertices of~$H$. For brevity, we write the edge $(x, \xi)$ simply as $x\xi$.}\label{fig:pp2dnfcode-d}
	\end{subfigure}\hfill\null%
	\caption{Example of the coding of a bipartite graph $H$ shown in Figure~\ref{fig:pp2dnfcode-a}.
		We encode $H$ relative to an instance $I_{e,\Pi}$ (Figure~\ref{fig:pp2dnfcode-b}), with a non-leaf
		edge~$e$ and an incident pair~$\Pi$.
		The result $I_{e,\Pi}^2$ of iterating $e$ in~$I$ with $n=2$ (Definition~\ref{def:iteration}) is shown in Figure~\ref{fig:pp2dnfcode-c}.
		The coding of~$H$ relative to~$I_{e,\Pi}$ and $n=2$  (Definition~\ref{def:ppcoding}) is shown in Figure~\ref{fig:pp2dnfcode-d}, with the probabilistic facts being the copies of~$F_{\ll}$ and $F_{\rr}$ in the edges in solid blue and black.
	}%
	\label{fig:pp2dnfcode}
\end{figure}
\begin{defi}%
	\label{def:ppcoding}
        Let $I_{e,\Pi}$ be a $\sigma$-instance where~$e = (u, v)$, $\Pi = (F_{\ll},F_{\rr})$, $F_{\ll} = R_{\ll}(l, u)$, $F_{\rr}=R_{\rr}(v, r)$, and let $n \geq 1$. Let $H = (A, B, C)$ be a connected bipartite graph.
        The \emph{coding} of~$H$ relative to $I_{e,\Pi}$ and~$n$  is a TID $\pdb = (J, \pi)$ with domain
        $\dom(J) \colonequals (\dom(I)\setminus \{u, v\}) \cup \{u_a \mid a \in A\} \cup \{v_b \mid b \in B\} \cup \{u_{c,2}, \ldots, u_{c,n} \mid c \in C\} \cup \{v_{c,1}, \ldots, v_{c,n-1} \mid c \in C\}$,
        where the new elements are fresh.
        The facts of the $\sigma$-instance~$J$ and the probability mapping~$\pi$ are defined as follows:
	\begin{itemize}
		\item \emph{Copy non-incident facts:} Initialize $J$ as the induced subinstance of~$I$ on $\dom(I) \setminus \{u, v\}$.
                \item \emph{Copy incident facts $F_{\ll}$ and $F_{\rr}$:}
                  Add to~$J$ the $\sigma^\leftrightarrow$-fact $R_{\ll}(l, u_a)$
                  for each $a \in A$, and add to~$J$ the $\sigma^\leftrightarrow$-fact $R_{\rr}(v_b, r)$
                  for each $b \in B$.
                \item \emph{Copy other left-incident facts:}
                  For each $\sigma^\leftrightarrow$-fact $F_{\ll}' = R_{\ll}'(l', u)$ of~$I$
                  that is left-incident to~$e$ (i.e., $l' \notin \{u, v\}$)
                  and where $F_{\ll}' \neq F_{\ll}$,
                  add to~$J$ the facts $R_{\ll}'(l', u_a)$ for each $a \in A$,
                  and add to~$J$ the facts $R_{\ll}'(l', u_{c,j})$
                  for each $2 \leq j \leq n$ and $c \in C$.
                \item \emph{Copy other right-incident facts:}
                  For each $\sigma^\leftrightarrow$-fact $F_{\rr}' = R_{\rr}'(v, r')$ of~$I$
                  that is right-incident to~$e$ (i.e., $r' \notin \{u, v\}$)
                  and where $F_{\rr}' \neq F_{\rr}$,
                  add to~$J$ the facts $R_{\rr}'(v_b, r')$ for each $b \in B$
                  and add to~$J$ the facts $R_{\rr}'(v_{c,j}, r')$
                  for each $1 \leq j \leq n-1$ and $c \in C$.
                \item \emph{Create copies of~$e$:} For each $c \in C$ with $c =
                  (a, b)$, copy~$e$ on the following pairs:
                  $(u_{c,i}, v_{c,i})$ for $1 \leq i \leq n$, and $(u_{c,i+1},
                  v_{c,i})$ for $1 \leq i \leq n-1$, where we use $u_{c,1}$ to
                  refer to $u_a$ and $v_{c,n}$ to refer to~$v_b$.
	\end{itemize}
        Finally, we define the function $\pi$ such that it maps all the facts created in the step ``Copy incident facts $F_{\ll}$ and $F_{\rr}$'' to $0.5$, and all other facts to $1$.
\end{defi}

Observe how this definition relates to the definition of iteration (Definition~\ref{def:iteration}): we intuitively code each edge of the bipartite graph as a copy of the path of copies of~$e$ in the definition of the $n$-th iterate of~$(u, v)$.
Note also that there are exactly $\card{A} + \card{B}$ uncertain facts, by construction. It is clear that, for any choice of $I_{e,\Pi}$ and~$n$, this coding is in polynomial time in~$H$.

We now define the bijection~$\phi$, mapping each possible world~$\omega$ of the connected bipartite graph~$H$ to a possible world of the TID $\pdb$. For each vertex $a \in A$, we keep the copy of~$F_{\rr}$ incident to~$u_a$ in $\phi(\omega)$ if~$a$ is kept in $\omega$, and we do not keep it otherwise; we do the same for~$v_b$, and~$F_{\ll}$.
It is obvious that this correspondence is bijective, and that all possible worlds have the same probability, namely, $0.5^{\card{A} + \card{B}}$.  Furthermore, we can use $\phi$ to define a reduction, thanks to the following statement:
 \begin{figure}[t]
   \drawnewkey%
   \bigskip

	\begin{subfigure}[t]{0.4\textwidth}
		\centering
		\begin{tikzpicture}[yscale=.5]
		\tikzstyle{varb} = [text width=1.1em, text centered, circle, draw,inner sep=0.5pt, blue]
		\tikzstyle{varr} = [text width=1em, text centered, circle, draw,inner sep=0pt, red]
		\node (A) at (0, 0.9) {$A$};
		\node (B) at (2, 0.9) {$B$};
		\node[varb] (a) at (0, 0) {$\alpha$};
		\node (b) at (0, -1) {$\beta$};
		\node (c) at (0, -2) {$\gamma$};
		\node (alpha) at (2, 0) {$\alpha'$};
		\node[varr]  (beta) at (2, -1) {$\beta'$};
		\node[varr]  (gamma) at (2, -2) {$\gamma'$};
		\draw (alpha) edge[-, thick, shorten >=1pt] (a);
		\draw (b) edge[-, thick] (alpha);
		\draw (b)  edge[-, thick, shorten >=1pt] (beta);
		\draw (c)  edge[-, thick, shorten >=1pt] (beta);
		\draw (c)  edge[-, thick, shorten >=1pt](gamma);
		\draw (c) edge[-,thick] (alpha);
		\end{tikzpicture}
                \caption{A possible world $\omega$ of $H$ from Figure~\ref{fig:pp2dnfcode-a}, containing all circled nodes.}\label{fig:pp2dnffold-a}
	\end{subfigure}
	\hspace*{1cm}
	\begin{subfigure}[t]{0.45\textwidth}
		\centering
		\begin{tikzpicture}[yscale=.5]
		\tikzstyle{varb} = [text width=1.1em, text centered, circle, draw,inner sep=0.5pt, blue]
		\tikzstyle{varr} = [text width=1em, text centered, circle, draw,inner sep=0pt, red]
		\node (A) at (0, 0.9) {$A$};
		\node (B) at (2, 0.9) {$B$};
		\node (A) at (4, 0.9) {$A$};
		\node (B) at (6, 0.9) {$B$};
		\node[varb]  (a) at (0, 0) {$\alpha$};
		\node (alpha) at (2, 0) {$\alpha'$};
		\node (b) at (4, 0) {$\beta$};
		\node (c) at (4, -1) {$\gamma$};
		\node (alpha2) at (6, 0) {$\alpha'$};
		\node[varr] (beta) at (6, -1) {$\beta'$};
		\node[varr] (gamma) at (6, -2) {$\gamma'$};
		\draw (alpha) edge[-, thick, shorten >=1pt] (a);
		\draw (alpha) edge[-, thick] (b);
		\draw (alpha) edge[-, thick] (c);
		\draw (b) edge[-, thick] (alpha2);
		\draw (b) edge[-, thick, shorten >=1pt] (beta);
		\draw (c) edge[-, thick, shorten >=1pt] (beta);
		\draw (c) edge[-, thick, shorten >=1pt] (gamma);
		\end{tikzpicture}
		\caption{The way $H$ is considered in the completeness proof of
                Proposition~\ref{prp:ppcoding}.}\label{fig:pp2dnffold-b}
	\end{subfigure}

	\bigskip
	\bigskip

	\null\hfill\begin{subfigure}[t]{.97\linewidth}
          \centering
          \begin{tikzpicture}[xscale=1.6,yscale=1.2]
		\node (v1) at (0, -0.5) {${l}$};
		\node (v2) at (0, -2) {${l_1}$};
		\node (ua) at (2.5, 1) {$\bm{u_\alpha}$};
		\node (uaalpha) at (2.6, 0.5) {$u_{\alpha\alpha',2}$};
		\node (ub) at (2.5, -0.5) {$\bm{u_\beta}$};
		\node (ubalpha) at (2.6, -1) {$u_{\beta\alpha',2}$};
		\node (ubbeta) at (2.6, -1.5) {$u_{\beta\beta',2}$};
		\node (uc) at (2.5, -2.5) {$\bm{u_\gamma}$};
		\node (ucalpha) at (2.6, -3) {$u_{\gamma\alpha',2}$};
		\node (ucbeta) at (2.6, -3.5) {$u_{\gamma\beta',2}$};
		\node (ucgamma) at (2.6, -4) {$u_{\gamma\gamma',2}$};
		\node (vaalpha) at (5.5, 1) {$v_{\alpha\alpha',1}$};
		\node (vbalpha) at (5.5, 0.5) {$v_{\beta\alpha',1}$};
		\node (vcalpha) at (5.5, 0) {$v_{\gamma\alpha',1}$};
		\node (valpha) at (5.6, -0.5) {$\bm{v_{\alpha'}}$};
		\node (vbbeta) at (5.5, -1.5) {$v_{\beta\beta',1}$};
		\node (vcbeta) at (5.5, -2) {$v_{\gamma\beta',1}$};
		\node (vbeta) at (5.6, -2.5) {$\bm{v_{\beta'}}$};
		\node (vcgamma) at (5.5, -3.5) {$v_{\gamma\gamma',1}$};
		\node (vgamma) at (5.6, -4) {$\bm{v_{\gamma'}}$};
		\node (u1) at (8, 0) {${r_1}$};
		\node (u2) at (8, -1.3) {${r}$};
		\node (u3) at (8, -2.5) {${r_2}$};
		\draw (v1) edge[-, thick, blue] (ua);
		\draw (v1) edge[-, thick, dashed,blue] (uaalpha);
		\draw (v2) edge[-, orange] (ua);
		\draw (v2) edge[-, orange] (uaalpha);
		\draw (v1) edge[-, thick, dashed, blue] (ub);
		\draw (v1) edge[-, thick, dashed,blue] (ubalpha);
		\draw (v1) edge[-, thick, dashed,blue] (ubbeta);
		\draw (v2) edge[-, orange] (ub);
		\draw (v2) edge[-, orange] (ubalpha);
		\draw (v2) edge[-, orange] (ubbeta);
		\draw (v1) edge[-, thick, dashed,blue] (uc);
		\draw (v1) edge[-, thick, dashed,blue] (ucalpha);
		\draw (v1) edge[-, thick, dashed,blue] (ucbeta);
		\draw (v1) edge[-, thick, dashed,blue] (ucgamma);
		\draw (v2) edge[-, orange] (uc);
		\draw (v2) edge[-, orange] (ucalpha);
		\draw (v2) edge[-, orange] (ucbeta);
		\draw (v2) edge[-, orange] (ucgamma);
		\draw (ua) edge[-, very thick, cyan] (vaalpha);
		\draw (uaalpha) edge[-, very thick,cyan] (vaalpha);
		\draw (uaalpha) edge[-, very thick,cyan] (valpha);
		\draw (ub) edge[-, very thick, cyan] (vbalpha);
		\draw (ubalpha) edge[-, very thick,cyan] (vbalpha);
		\draw (ubalpha) edge[-, very thick,cyan] (valpha);
		\draw (ub) edge[-, very thick, cyan] (vbbeta);
		\draw (ubbeta) edge[-, very thick,cyan] (vbbeta);
		\draw (ubbeta) edge[-, very thick,cyan] (vbeta);
		\draw (uc) edge[-, very thick, cyan] (vcbeta);
		\draw (ucbeta) edge[-, very thick,cyan] (vcbeta);
		\draw (ucbeta) edge[-, very thick,cyan] (vbeta);
		\draw (uc) edge[-, very thick, cyan] (vcgamma);
		\draw (ucgamma) edge[-, very thick,cyan] (vcgamma);
		\draw (ucgamma) edge[-, very thick,cyan] (vgamma);
		\draw (uc) edge[-, very thick, cyan] (vcalpha);
		\draw (ucalpha) edge[-, very thick,cyan] (vcalpha);
		\draw (ucalpha) edge[-, very thick,cyan] (valpha);
		\draw (vaalpha) edge[-, orange] (u1);
		\draw (vaalpha) edge[-, thick,dashed, black] (u2);
		\draw (vaalpha) edge[-, orange] (u3);
		\draw (vbalpha) edge[-, orange] (u1);
		\draw (vbalpha) edge[-, thick,dashed, black] (u2);
		\draw (vbalpha) edge[-, orange] (u3);
		\draw (vcalpha) edge[-, orange] (u1);
		\draw (vcalpha) edge[-, thick,dashed, black] (u2);
		\draw (vcalpha) edge[-, orange] (u3);
		\draw (valpha) edge[-, orange] (u1);
		\draw (valpha) edge[-, thick,dashed,black] (u2);
		\draw (valpha) edge[-, orange] (u3);
		\draw (vbbeta) edge[-, orange] (u1);
		\draw (vbbeta) edge[-, thick,dashed, black] (u2);
		\draw (vbbeta) edge[-, orange] (u3);
		\draw (vcbeta) edge[-, orange] (u1);
		\draw (vcbeta) edge[-, thick,dashed, black] (u2);
		\draw (vcbeta) edge[-, orange] (u3);
		\draw (vbeta) edge[-, orange] (u1);
		\draw (vbeta) edge[-, thick,black] (u2);
		\draw (vbeta) edge[-, orange] (u3);
		\draw (vcgamma) edge[-, orange] (u1);
		\draw (vcgamma) edge[-, thick,dashed, black] (u2);
		\draw (vcgamma) edge[-, orange] (u3);
		\draw (vgamma) edge[-, orange] (u1);
		\draw (vgamma) edge[-, thick,black] (u2);
		\draw (vgamma) edge[-, orange] (u3);
		\end{tikzpicture}
                \caption{The possible world $\phi(\omega)$ of the coding (Figure~\ref{fig:pp2dnfcode-d}) for~$\omega$.
                 The edges $(l,u_\beta)$, $(l,u_\gamma)$, and $(v_{\alpha'}, r)$ are changed to dashed lines, as they correspond to vertices of~$H$ that are not kept in~$\omega$.
		}\label{fig:pp2dnffold-c}
	\end{subfigure}\hfill\null%

	\caption{
		Example for the completeness direction of the proof of Proposition~\ref{prp:ppcoding}.
		Figure~\ref{fig:pp2dnffold-a} shows a bad possible world~$\omega$ of the bipartite graph.
         The corresponding possible world of the coding of Figure~\ref{fig:pp2dnfcode-d} (using the instance $I_{e,\Pi}^2$ of Figure~\ref{fig:pp2dnfcode-b}) is given in Figure~\ref{fig:pp2dnffold-c}.
         In the proof, we explore $H$ as depicted in Figure~\ref{fig:pp2dnffold-b} to argue that Figure~\ref{fig:pp2dnffold-c} has a homomorphism to~$I_{e,\Pi}^5$.
	}%
	\label{fig:pp2dnffold}
\end{figure}
\begin{prop}%
	\label{prp:ppcoding}
	 Let the TID $\pdb = (J, \pi)$ be the coding of a connected bipartite graph $H=(A,B,C)$ relative to an instance $I_{e,\Pi}$ and to $n \geq 1$ as described in Definition~\ref{def:ppcoding}, and let $\phi$ be the bijective function defined above from the possible worlds of $H$ to those of $\pdb$. Then:
	 \begin{enumerate}
        \item For any \emph{good} possible world $\omega$ of~$H$, $\phi(\omega)$ has a homomorphism \emph{from}~$I_{e,\Pi}^{n}$.
		\item For any \emph{bad} possible world $\omega$ of~$H$, $\phi(\omega)$ has a homomorphism \emph{to} $I_{e,\Pi}^{3n-1}$.
        \end{enumerate}
\end{prop}

\begin{proof}
 Observe that (1) corresponds to the soundness of the reduction, and (2) to the completeness. Intuitively, (1) holds because $\phi(\omega)$ then contains a subinstance isomorphic to~$I_{e,\Pi}^n$. To show (2), we need a more involved argument, see e.g.~Figure~\ref{fig:pp2dnffold}: when $\omega$ is bad, we can show how to ``fold back'' $\phi(\omega)$, going from the copies of $F_{\ll}$ to the copies of $F_{\rr}$, into the iterate $I_{e,\Pi}^{3n-1}$. This uses the fact that $\omega$ is bad, so the copies of~$F_{\ll}$ and $F_{\rr}$ must be sufficiently far from one another.

\begin{enumerate}
\item  Let us assume that $\phi(\omega)=J'$. We more specifically claim that~$J'$ has a subinstance which is isomorphic to~$I_{e,\Pi}^n$. To see why, drop all copies of~$u$ from~$J'$ except $u_a$ and the $u_{c,i}$, and all copies of $v$ except $v_b$ and the $v_{c,i}$, along with all facts where these elements appear. All of the original instance~$I$ except for the facts involving~$u$ and~$v$ can be found as-is in~$J'$. Now, for the others, $u_a$ has an incident copy of all edges incident to~$u$ in~$J'$ (including~$F_{\ll}$), the same is true for~$v_b$ and~$v$ (including~$F_{\rr}$), and we can use the $u_{e,i}$ and $v_{e,i}$ to witness the requisite path of copies of~$e$.
\item 	As before, let us assume that $\phi(\omega)=J'$. Let us describe the homomorphism from~$J'$ to~$I_{e,\Pi}^{3n-1}$. To do this, first map all facts of~$J'$ that do not involve a copy of~$u$ or~$v$ to the corresponding facts of~$I_{e,\Pi}^{3n-1}$ using the identity mapping. We will now explain how the copies of $u$ and $v$ in~$J'$ are mapped to copies of~$u$ and~$v$ in~$I_{e,\Pi}^{3n-1}$: this is clearly correct for the facts in~$J'$ that use these copies of~$u$ and~$v$ and that were created as copies of left-incident or right-incident facts to~$e$ in~$I$ except $F_{\ll}$ and~$F_{\rr}$. Thus, we must simply ensure that this mapping respects the facts in~$J'$ that were created as copies of~$F_{\ll}$, of~$F_{\rr}$, or of the edge~$e$, as we have argued that all other facts of~$J'$ are correctly mapped to~$I_{e,\Pi}^{3n-1}$.

 Our way to do this is illustrated in Figure~\ref{fig:pp2dnffold}.
    The first step is to take all copies of~$F_{\ll}$ in~$J'$, which correspond to vertices in~$a \in A$ that were kept, and to map them all to the element~$u$ in~$I_{e,\Pi}^{3n-1}$, which is possible as it has the incident fact~$F_{\ll}$. In Figure~\ref{fig:pp2dnffold-c}, this is only the copy of~$F_{\ll}$ on $(l, u_{\alpha})$.
    Now, we start at the elements of the form~$u_a$, and we follow the paths of $2n-1$ copies of~$e$ back-and-forth from these elements until we reach elements of the form~$v_b$: we map these paths to the first $2n-1$ edges of the path of copies of~$e$ from~$u$ to~$v$ in~$I_{e,\Pi}^{3n-1}$. In Figure~\ref{fig:pp2dnffold-c}, we reach $v_{\alpha'}$.
  From our assumption about the possible world~$J'$, none of the $v_b$ reached at that stage have an incident copy of~$F_{\rr}$, as we would otherwise have a witness to the fact that we kept two adjacent $a\in A$ and $b\in B$ in the possible world~$\omega$ of~$H$, which is impossible as~$\omega$ is bad.

  The second step is to go \emph{back} in~$J'$ on the copies of~$e$ incident to these elements that were not yet visited, and we follow a path of copies of~$e$ that were not yet mapped. We map these to the next $2n-1$ copies of~$e$, going \emph{forward} in the path from~$u$ to~$v$ in~$I_{e,\Pi}^{3n-1}$. We then reach elements of the form~$u_a$, and they do not have any incident~copies of~$F_{\ll}$ because all such edges and their
  outgoing paths were visited in the first step. In Figure~\ref{fig:pp2dnffold-c}, we reach~$u_\beta$ and~$u_\gamma$.

    The third step is to go \emph{forward} in~$J'$ on the copies of~$e$ incident to these elements that were not yet visited, and follow a path of copies of~$e$ that goes to elements of the form~$v_b$, mapping this to the last $2n-1$ edges of the path from~$u$ to~$v$ in~$I_{e,\Pi}^{3n-1}$. Some of these $v_b$ may now be incident to copies of~$F_{\rr}$, but the same is true of~$v$ in~$I_{e,\Pi}^{3n-1}$, and we have just reached~$v$. Indeed, note that we have followed $(2n-1)\times 3$ copies of~$e$ in~$J'$ in total (going forward each time), and this is equal to $2\times(3n-1)-1$, the number of copies of~$e$ in the path from~$u$ to~$v$ in~$I_{e,\Pi}^{3n-1}$. Thus, we can map the copies of~$F_{\rr}$ correctly. In Figure~\ref{fig:pp2dnffold-c}, we reach~$v_{\beta'}$ and~$v_{\gamma'}$.

  In Figure~\ref{fig:pp2dnffold-c}, we have visited everything after the third step. However, in general, there may be some elements of~$J'$ that we have not yet visited, and for which we still need to define a homomorphic image. Thus, we perform more steps until all elements are visited. Specifically, in even steps,
    we go \emph{back} on copies of~$e$ in~$J'$ from the elements reached in the previous step to reach elements that were not yet visited, going \emph{back} on the path from~$u$ to~$v$ in~$I_{e,\Pi}^{3n-1}$, reaching elements of the form $u_a$ (which cannot be incident to any copy of~$F_{\ll}$ for the same reason as in the second step). In odd steps, we go \emph{forward} in~$J'$ on copies of~$e$, going \emph{forward} on the path from~$u$ to~$v$ in~$I_{e,\Pi}^{3n-1}$, reaching elements of the form~$v_b$ in~$J'$ that we map to~$b$ in~$I_{e,\Pi}^{3n-1}$, including the~$F_{\rr}$-fact that may be incident to them.

    We repeat these additional backward-and-forward steps until everything reachable has been visited.
  At the end of the process, from our assumption that $H$ is a connected bipartite graph, we have visited all the elements of~$J'$ for which we had not defined a homomorphic image yet, and we have argued that the way we have mapped them is indeed a homomorphism.
  This concludes the construction of the homomorphism, and concludes the proof. \qedhere
\end{enumerate}
\end{proof}

\noindent
Thanks to Proposition~\ref{prp:ppcoding}, we can now prove the main result of this section:
\begin{proof}[Proof of Theorem~\ref{thm:pp2dnfred}]
  Fix the query $Q$, the instance $I$, the non-leaf edge~$e$ of~$I$ which is non-iterable, the incident pair $\Pi$ relative to which it is not iterable, and let us take the smallest $n_0 > 1$ such that $I_{e,\Pi}^{n_0}$ does not satisfy the query, but $I_{e,\Pi}^{n_0-1}$ does.

   We show the \#P-hardness of~$\PQE(Q)$ by reducing from \pptwodnf (Definition~\ref{def:pp2dnf}). Let $H = (A, B, C)$ be an input connected bipartite graph. We apply the coding of Definition~\ref{def:ppcoding} with~$n_0-1$ and obtain a TID $\pdb$. This coding can be done in polynomial time.

    Now let us use Proposition~\ref{prp:ppcoding}. We know that $I_{e,\Pi}^{n_0-1}$ satisfies~$Q$, but $I_{e,\Pi}^{3(n_0-1)-1}$ does not, because $n_0 > 1$ so $3(n_0-1)-1=3n_0-4 \geq n_0$, and as we know that $I_{e,\Pi}^{n_0}$ violates~$Q$, then so does $I_{e,\Pi}^{3(n_0-1)-1}$ thanks to Observation~\ref{obs:iterhomom}.
	Thus, Proposition~\ref{prp:ppcoding} implies that the number of good possible worlds of~$H$ is the probability that~$Q$ is satisfied in a possible world of~$\pdb$, multiplied by the constant factor~$2^{\card{A} + \card{B}}$. Thus, the number of good possible worlds of~$H$ is $\Pr_\pdb(Q) \cdot 2^{\card{A} + \card{B}}$. This shows that the reduction is correct, and concludes the proof.
\end{proof}

\section{Finding a Minimal Tight Pattern}%
\label{sec:findhard}

In the previous section, we have shown hardness for queries (bounded or unbounded) that have a model with a non-iterable, non-leaf edge. This leaves open the case of unbounded queries for which all non-leaf edges in all models can be iterated.
We first note that this case is not hypothetical, i.e., there actually exist some unbounded queries for which, in all models, all non-leaf edges can be iterated:
\begin{exa}%
  \label{exa:badq}
Consider the following Datalog program:
\begin{align*}
    R(x, y) &\rightarrow A(y), \\
   A(x), S(x, y) &\rightarrow B(y), \\
   B(x), S(y, x) &\rightarrow A(y), \\
   B(x), T(x, y) &\rightarrow \mathrm{Goal}().
\end{align*}
This program is unbounded, as it tests if the instance contains a path of the form $R(a, a_1), \allowbreak S(a_1, a_2), S^-(a_2, a_3), \ldots, S(a_{2n+1}, a_{2n+2}), T(a_{2n+2}, b)$.
However, it has no model with a non-iterable, non-leaf edge: in every model, the query is satisfied by a path of the form above, and we cannot break such a path by iterating a non-leaf edge (i.e., this yields a longer path of the same form).
\end{exa}

Importantly, if we tried to reduce from \pptwodnf for this query as in the previous section, then the reduction would fail because the edge is iterable: in possible worlds of the bipartite graph, where we have not retained two adjacent vertices, we would still have matches of the query in the corresponding possible world of the probabilistic instance, where we go from a chosen vertex to another by going \emph{back-and-forth} on the copies of~$e$ that code the edges of the bipartite graph. These are the ``back-and-forth matches'' which were missed in~\cite{jung2014reasoning,JL12} and are discussed in~\cite{JL20}.

In light of this, we handle the case of such queries in the next two sections. In this section, we prove a general result for unbounded queries (independent from the previous section): all unbounded queries must have a model with a \emph{tight edge}, which is additionally \emph{minimal} in some sense. Tight edges and iterable edges will then be used in Section~\ref{sec:ustcon} to show hardness for unbounded queries which are not covered by the previous section.

Let us start by defining this notion of \emph{tight edge}, via a rewriting operation on instances called a \emph{dissociation}.

\begin{defi}%
  \label{def:coarsedissociation}
    The \emph{dissociation} of a non-leaf edge $e = (u, v)$ in~$I$ is the instance $I'$ where:
	\begin{itemize}
		\item $\dom(I') = \dom(I) \cup \{u', v'\}$ where $u'$ and $v'$ are fresh.
         \item $I'$ is $I$ where we create a copy of the edge $e$ on~$(u, v')$ and on~$(u', v)$, and then remove all non-unary facts covered by~$e$ in~$I'$.
	\end{itemize}
\end{defi}

\noindent
Dissociation is illustrated  in the following example (see also Figure~\ref{fig:dissociation}):

\begin{exa}
  Consider the following instance:
  \[
  I = \{R(a, b), S(b, a), T(b, a), R(a, c), S(c, b), S(d, b), U(a, a), U(b, b)\}.
  \]
  The edge $(a, b)$ is non-leaf, as witnessed by the edges $\{a, c\}$ and $\{b, c\}$.
  The result of the dissociation is then:
  \begin{align*}
     I' = \{ &R(a, b'), S(b', a), T(b', a), R(a', b), S(b, a'), T(b, a'),\\ &R(a, c), S(c, b), S(d, b), U(a, a), U(a', a'), U(b, b), U(b', b')\}
  \end{align*}
\end{exa}

We then call an edge \emph{tight} in a model of~$Q$ if dissociating it makes~$Q$ false:

\begin{defi}
  Let $Q$ be a query and $I$ be a model of~$Q$. An edge $e$ of~$I$ is \emph{tight} if it is non-leaf, and the result of the dissociation of~$e$ in~$I$ does not satisfy~$Q$.
  A \emph{tight pattern} for the query~$Q$ is a pair $(I, e)$ of a model $I$ of~$Q$ and of an edge $e$ of~$I$ that is tight.
\end{defi}

Intuitively, a tight pattern is a model of a query containing at least three
edges $\{u, a\}, \allowbreak \{a, b\},\allowbreak  \{b, v\}$ (possibly $u = v$) such that performing a dissociation makes the query false.  For instance, for the unsafe CQ $Q_0: R(w, x), S(x,y), T(y, z)$ from~\cite{DaSu07}, a tight pattern would be $\{R(a, b), S(b, c), T(c, d)\}$ with the edge $(b, c)$. Again, not all unsafe CQs have a tight pattern, e.g., ${Q_0': R(x, x), S(x, y), T(y, y)}$, and ${Q_1: (R(w, x), S(x, y)) \lor (S(x, y), T(y, z))}$ from Section~\ref{sec:pp2dnf} do not.

For our purposes, we will not only need tight patterns, but \emph{minimal tight patterns}:

\begin{defi}%
	\label{def:minimal}
        Given an instance $I$ with a non-leaf edge $e = (a, b)$, the \emph{weight} of~$e$  is the number of facts covered by~$e$ in~$I$ (including unary facts). The \emph{side weight} of~$e$ is the number of
        $\sigma^\leftrightarrow$-facts
        in~$I$ that are left-incident to~$e$, plus the number of $\sigma^\leftrightarrow$-facts in~$I$ that are right-incident\footnote{Recall that left-incident and right-incident facts do not include unary facts.} to~$e$.
        Given a query~$Q$, we say that a tight pattern $(I, e)$ is \emph{minimal} if:
	\begin{itemize}
		\item $Q$ has no tight pattern $(I', e')$ where the weight of~$e'$ is
		strictly less than that of~$e$; and
		\item $Q$ has no tight pattern $(I', e')$ where the weight of~$e'$ is
		equal to that of~$e$ and the side weight of~$e'$ is strictly less than
		that of~$e$.
	\end{itemize}
\end{defi}

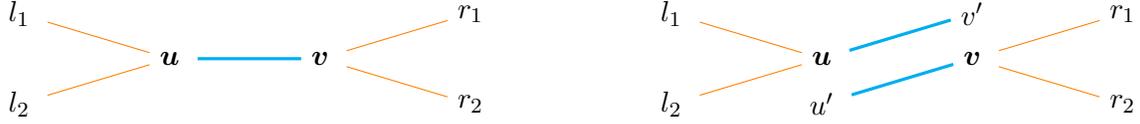
\begin{figure}[t]
	\begin{tikzpicture}[node distance = 0.5cm,shorten <=3pt,yscale=.6]
	\node (v1) at (0, 0) {$l_1$};
	\node (v2) at (0, -2) {$l_2$};
	\node (u) at (2, -1) {$\bm{u}$};
	\node (v) at (4, -1) {$\bm{v}$};
	\node (u1) at (6, -2) {$r_2$};
	\node (u2) at (6, 0) {$r_1$};
	\draw (u) edge[very thick,cyan] (v);
	\draw (v1) edge[orange] (u);
	\draw (v2) edge[orange] (u);
	\draw (v) edge[orange] (u1);
	\draw (v) edge[orange] (u2);
	\end{tikzpicture}
	\hfill
	\begin{tikzpicture}[node distance = 0.5cm,shorten <=3pt,yscale=.6]
	\node (v1) at (0, 0) {$l_1$};
	\node (v2) at (0, -2) {$l_2$};
	\node (u) at (2, -1) {$\bm{u}$};
	\node (up) at (2, -2) {$u'$};
	\node (v) at (4, -1) {$\bm{v}$};
	\node (vp) at (4, 0) {$v'$};
	\node (u1) at (6, 0) {$r_1$};
	\node (u2) at (6, -2) {$r_2$};
	\draw (v1) edge[orange] (u);
	\draw (v2) edge[orange] (u);
	\draw (u) edge[very thick,cyan] (vp);
	\draw (up) edge[very thick,cyan] (v);
	\draw (v) edge[orange] (u1);
	\draw (v) edge[orange] (u2);
	\end{tikzpicture}
        \caption{An instance (left) with a non-leaf edge $(u, v)$, and the
        result (right) of dissociating $(u, v)$.}%
	\label{fig:dissociation}
\end{figure}

\noindent
We can now state the main result of this section:

\begin{thm}%
	\label{thm:findhard2}
	Every unbounded \ucqinf~$Q$ has a model $I$ with a non-leaf edge $e$ such that $(I, e)$ is a minimal tight pattern.
\end{thm}

The idea of how to find tight patterns is as follows. We first note that the only instances without non-leaf edges are intuitively disjoint unions of star-shaped subinstances.
Now, if a query is unbounded, then its validity cannot be determined simply by looking at such subinstances (unlike $Q'_0$ or~$Q_1$  from Section~\ref{sec:pp2dnf}), so there must be a model of the query with an edge that we cannot dissociate without breaking the query, i.e., a tight pattern. Once we know that there is a tight pattern, then it is simple to argue that we can find a model with a tight edge that is minimal in the sense that we require.

To formalize this intuition, let us first note that any \emph{iterative dissociation process}, i.e., any process of iteratively applying dissociation to a given instance, will necessarily terminate. More precisely, an \emph{iterative dissociation process} is a sequence of instances starting at an instance~$I$ and where each instance is defined from the previous one by performing the dissociation of some non-leaf edge. We say
that the process \emph{terminates} if it reaches an instance where there is no edge left to dissociate, i.e., all edges are leaf edges.
\begin{obs}%
\label{obs:nonleaf}
  For any instance $I$, any iterative dissociation process will terminate in $n$ steps, where $n$ is the number of non-leaf edges in~$I$.
\end{obs}

\begin{proof}
It is sufficient to show that an application of dissociation decreases the number of non-leaf edges by $1$. To do so, we consider an instance $I$ with a non-leaf edge $e$, and show that the dissociation~$I'$ of~$e$ in~$I$, has $n-1$ non-leaf edges.

Let us write $e = (a, b)$. The new elements $a'$ and $b'$ in~$I'$ are leaf elements, and for any other element of the domain of~$I'$, it is a leaf in~$I'$ iff it was a leaf in~$I$: this is clear for elements that are not $a$ and~$b$ as they occur exactly in the same edges, and for $a$ and~$b$ we know that they were not leaves in~$I$  (they occurred in $e = \{a, b\}$ and in some other edge), and they are still not leaves in~$I'$ (they occur in the same other edge and in $\{a, b'\}$ and $\{b, a'\}$, respectively).

Thus, the edges of~$I'$ that are not $\{a, b'\}$ or $\{a', b\}$ are leaf edges in~$I'$ iff they were in~$I$. So, in terms of non-leaf edges the only difference between $I$ and $I'$ is that we removed the non-leaf edge $\{a, b\}$ from $I$ and we added the two edges $\{a, b'\}$ and $\{a',b\}$ in~$I'$ which are leaf edges because $a'$ and~$b'$ are leaves. Thus, we conclude the claim.
\end{proof}

Let us now consider instances with no non-leaf edges. As we explained, they are intuitively disjoint unions of star-shaped subinstances, and in particular they homomorphically map to some constant-sized subset of their facts,
as will be crucial when studying our unbounded query.

\begin{prop}%
	\label{prp:stars}
	For every signature $\sigma$, there exists a bound $k_\sigma > 0$,  ensuring the following: for every instance $I$  on~$\sigma$ having no non-leaf edge, there exists a subinstance $I' \subseteq I$ such that $I$ has a homomorphism to~$I'$ and such that we have $\card{I'}<k_\sigma$.
\end{prop}

\begin{proof}
  We start by outlining the main idea behind of the proof. Connected instances having no non-leaf edges can have at most one non-leaf element, with all edges using this element and a leaf. Now, each edge can be described by the set of facts that it covers, for which there are finitely many possibilities (exponentially many in the signature size). We can thus collapse together the edges that have the same set of facts and obtain the subinstance~$I'$.
  Now, disconnected instances having no non-leaf edges are unions of the connected instances of the form above, so the number of possibilities up to homomorphic equivalence is again finite (exponential in the number of possible connected instances). We can then conclude by collapsing together connected components that are isomorphic.

  Let us now formally prove the result, first for connected instances $I$. In this case, we define the constant $k_\sigma' \colonequals 2^{4 \times \card{\sigma}}$.  There are two cases. The first case is when all elements of~$I$ are leaves: then, as $I$ is connected, it must consist of a single edge $(a, b)$ and consists of at most $4 \card{\sigma}$ facts: there are $\card{\sigma}$ possible facts of the form $R(a, b)$, plus $\card{\sigma}$ possible facts of the form $R(b, a)$, plus $\card{\sigma}$ possible facts of the form $R(a, a)$, plus $\card{\sigma}$ possible facts of the form $R(b, b)$. Thus, taking $I' = I$ and the identity homomorphism concludes the proof.
  The second case is when $I$ contains a non-leaf element~$a$. In this case, consider all edges $\{a, b_1\}, \ldots, \{a, b_n\}$ incident to~$a$.
  Each of the $b_i$ must be leaves: if some $b_i$ is not a leaf then $\{a, b_i\}$ would be a non-leaf edge because neither $a$ nor~$b_i$ would be leaves, contradicting our assumption that $I$ has no non-leaf edge.
  We then define an equivalence relation~$\sim$ on the $b_i$ by writing $b_i \sim b_j$ if the edges $\{a, b_i\}$ and $\{a, b_j\}$ contain the exact same set of facts (up to the isomorphism mapping $b_i$ to $b_j$): there are at most $k_\sigma'$ equivalence classes. The requisite subset of~$I$ and the homomorphism can thus be obtained by picking one representative of each equivalence class, keeping the edges incident to these representatives, and mapping each $b_i$ to the chosen representative of its class.

  Second, we formally show the result for instances $I$ that are not necessarily connected. Letting $I$ be such an instance, we consider its connected components $I_1, \ldots, I_m$, i.e., the disjoint subinstances induced by the connected components of the Gaifman graph of~$I$.
  Each of these is connected and has no non-leaf edges, so, by the proof of the previous paragraph,
  there are subsets $I_1' \subseteq I_1, \ldots, I_m' \subseteq I_m$ with at most $k_\sigma'$ facts each and a homomorphism of each $I_i$ to its $I_i'$.
  Now, there are only constantly many instances with at most $k_\sigma'$ facts up to isomorphism: let $k_\sigma''$ be their number, and let $k_\sigma \colonequals k_\sigma'' \times k_\sigma'$. The requisite subinstance and homomorphism is obtained by again picking one representative for each isomorphism equivalence class of the $I_i'$ (at most $k_\sigma''$ of them, so at most $k_\sigma$ facts in total) and mapping each $I_i$ to the $I_j'$ which is the representative for its~$I_i'$. This concludes the proof.
\end{proof}

We can now prove Theorem~\ref{thm:findhard2} by appealing to the unboundedness of the query.
To do this, we will rephrase unboundedness in terms of \emph{minimal models}:

\begin{defi}
A \emph{minimal model} of a query~$Q$ is an instance~$I$ that satisfies~$Q$ and such that every proper subinstance of~$I$ violates~$Q$.
\end{defi}

We can rephrase the unboundedness of a \ucqinf $Q$ in terms of minimal models: $Q$ is unbounded iff it has infinitely many minimal models. Indeed, if a query $Q$ has finitely many minimal models, then it is clearly equivalent to the UCQ formed from these minimal models, because it is closed under homomorphisms.
Conversely, if $Q$ is equivalent to a UCQ, then it has finitely many minimal models which are obtained as homomorphic images of  the UCQ disjuncts. Thus, we can clearly rephrase unboundedness as follows:
\begin{obs}%
  \label{obs:infmod}
  A \ucqinf query $Q$ is unbounded iff it has a minimal model $I$ with more
  than~$k$ facts for any $k \in \NN$.
\end{obs}

We are now ready to conclude the proof of Theorem~\ref{thm:findhard2}:
\begin{proof}[Proof of Theorem~\ref{thm:findhard2}]
We start by showing the first part of the claim: any unbounded query has a tight pattern. Let $k_\sigma$ be the bound from Proposition~\ref{prp:stars}. By Observation~\ref{obs:infmod}, let $I_0$ be a minimal model with more than $k_\sigma$ facts. Set $I \colonequals I_0$ and let us apply an iterative dissociation process: while $I$ has edges that are non-leaf but not tight, perform the dissociation, yielding $I'$, and let $I \colonequals I'$.

Observation~\ref{obs:nonleaf} implies that the dissociation process must terminate after at most $n_0$ steps, where $n_0$ is the number of non-leaf edges of~$I_0$. Let $I_n$ be the result of this process, with $n \leq n_0$. If $I_n$ has a non-leaf edge $e$ which is tight, then we are done as we have found a tight pattern $(I, e)$. Otherwise, let us reach a contradiction.

First notice that, throughout the rewriting process, it has remained true that $I$ is a model of~$Q$. Indeed, if performing a dissociation breaks this, then the dissociated edge was tight. Also notice that, throughout the rewriting, it has remained true that $I$ has a homomorphism to~$I_0$: it is true initially, with the identity homomorphism, and when we dissociate $I$ to~$I'$ then $I'$ has a homomorphism to~$I$ defined by mapping the fresh elements $a'$ and $b'$ to the original elements $a$ and~$b$ and as the identity otherwise. Hence, $I_n$ is a model of~$Q$ having a homomorphism to~$I_0$.

Note that $I_n$ has no non-leaf edges. Thus, Proposition~\ref{prp:stars} tells us that $I_n$ admits a homomorphism to some subset $I_n'$ of size at most~$k_\sigma$. This homomorphism witnesses that $I_n'$ also satisfies~$Q$. But now, $I_n'$ is a subset of~$I_n$ so it has a homomorphism to~$I_n$, which has a homomorphism to~$I_0$.
Let $I_0' \subseteq I_0$ be the image of~$I_n'$ in~$I_0$ by the composed homomorphism. It has at most $k_\sigma$ facts, because $I_n'$ does; and it satisfies $Q$ because $I_n'$ does. But as $I_0$ had more than $k_\sigma$ facts, $I_0'$ is a strict subset of $I_0$ that satisfies $Q$. This contradicts the minimality of~$I_0$. Thus, we conclude the first part of the claim.

It only remains to show the second part of the claim: there exists a minimal tight pattern. We already concluded that $Q$ has a tight pattern $(I, e)$, and $e$ has some finite weight $w_1 > 0$ in~$I$. Pick the minimal $0 < w_1' \leq w_1$ such that $Q$ has a tight pattern $(I', e')$ where $e'$ has weight $w_1'$. Now, $e'$ has some finite side weight $w_2 \geq 2$ in~$I'$. Pick the minimal $2 \leq w_2' \leq w_2$ such that $Q$ has a tight pattern $(I'', e'')$, where $e'$ has weight $w_1'$ and has side weight $w_2'$. We can then see that $(I'', e'')$ is a minimal tight pattern by minimality of~$w_1'$ and~$w_2'$. This concludes the proof.
\end{proof}

\section{Hardness with Tight Iterable Edges}%
\label{sec:ustcon}
In this section, we conclude the proof of Theorem~\ref{thm:main2} by showing that a minimal tight pattern can be used to show hardness when it is iterable. Formally:

\begin{thm}%
  \label{thm:ustcon}
  For every \ucqinf~$Q$, if $Q$ has a model $I$ with a non-leaf edge $e$ that is iterable
  then $\PQE(Q)$ is \#P-hard.
\end{thm}

This covers all the queries to which Section~\ref{sec:pp2dnf} did not apply. We note however that it does not subsume the result of that section, i.e., there are some unbounded queries to which it does not apply.

\begin{exa}
  Consider again the RPQ $R S^* T$ from Example~\ref{exa:rstrpq}. Recall that we can find some models with iterable edges (e.g., $\{R(a, b), S(b, c), T(c, d), R(a', b'), S(b', c'), T(c', d')\}$), but this query has no models with an iterable edge which is tight. Thus, hardness for this query cannot be shown with the result in this section, and we really need Theorem~\ref{thm:pp2dnfred} to cover it. Of course, there are also some unbounded queries for which hardness can be shown with either of the two results, e.g., a disjunction of the RPQ $R S^* T$ and of the query of Example~\ref{exa:badq} on a disjoint signature.
\end{exa}

From Theorem~\ref{thm:ustcon}, it is easy to conclude the proof of Theorem~\ref{thm:main2}:

\begin{proof}[Proof of Theorem~\ref{thm:main2}]
  Let $Q$ be an unbounded \ucqinf.
  If we have a model of~$Q$ with a non-iterable edge, then we conclude by
  Theorem~\ref{thm:pp2dnfred} that $\PQE(Q)$ is \#P-hard.
  Otherwise, by Theorem~\ref{thm:findhard2}, we have a minimal tight pattern,
  and its edge is then iterable (otherwise the first case would have applied),
  so that we can apply Theorem~\ref{thm:ustcon}.
\end{proof}

Thus, it only remains to show Theorem~\ref{thm:ustcon}. The idea is to use the iterable edge~$e$ of the minimal tight pattern $(I, e)$ for some incident pair $\Pi$ to reduce from the undirected st-connectivity problem \stcon (Definition~\ref{def:stcon}).
Given an input st-graph~$G$ for \stcon, we will code it as a TID $\pdb$ built using~$I_{e,\Pi}$, with one probabilistic fact per edge of~$G$. To show a reduction, we will argue that good possible worlds of~$G$ correspond to possible worlds $J'$ of~$\pdb$ containing some iterate $I_{e,\Pi}^n$ of the instance (Definition~\ref{def:iteration}), with $n$ being the length of the path, and $J'$ then satisfies~$Q$ because~$e$ is iterable.
Conversely, we will argue that bad possible worlds of~$G$ correspond to possible worlds $J'$ of~$\pdb$ that have a homomorphism to a so-called \emph{fine dissociation} of~$e$ in~$I$, and we will argue that this violates the query~$Q$ thanks to our choice of~$(I, e)$ as a minimal tight pattern.
The notion of \emph{fine dissociation} will be defined for an edge relative to an incident pair, but also relative to a specific choice of fact covered by the edge, as we formally define below (and illustrate in Figure~\ref{fig:finedissociation}):

\begin{defi}%
  \label{def:finedissociation}
  Let $I$ be a $\sigma$-instance, let $e = (u, v)$ be a non-leaf edge in~$I$, let $F_{\ll} = R_{\ll}(l, u)$ and $F_{\rr} = R_{\rr}(v, r)$ be an incident pair of~$e$ in~$I$, and let $F_{\mm}$ be a non-unary fact covered by the edge~$e$.
  The result of performing the \emph{fine dissociation} of~$e$ in~$I$ relative to~$F_{\ll},F_{\rr}$ and $F_{\mm}$ is a $\sigma$-instance $I'$ on the domain $\dom(I')=\dom(I) \cup \{u', v'\}$, where the new elements are fresh.
  It is obtained by applying the following steps:
  \begin{itemize}
    \item \emph{Copy non-incident facts:} Initialize $I'$ as the induced subinstance of~$I$ on $\dom(I) \setminus \{u, v\}$.
    \item \emph{Copy incident facts $F_{\ll}$ and $F_{\rr}$:} Add the facts $F_{\ll}$ and $F_{\rr}$ to~$I'$.
    \item \emph{Copy other left-incident facts:} For every $\sigma^\leftrightarrow$-fact $F_{\ll}' = R_{\ll}'(l', u)$
      of~$I$ that is left-incident to~$e$ (i.e., $l' \notin \{u, v\}$)
      and where $F_{\ll}' \neq F_{\ll}$,
      add to~$I'$ the fact $R_{\ll}'(l', u')$.
    \item \emph{Copy other right-incident facts:}
      For every $\sigma^\leftrightarrow$-fact $F_{\rr}' = R_{\rr}'(v, r')$
      of~$I$ that is right-incident to~$e$ (i.e., $r' \notin \{u, v\}$)
      and where $F_{\rr}' \neq F_{\rr}$,
      add to~$I'$ the fact $R_{\rr}'(v', r')$.
    \item \emph{Create the copies of~$e$:} Copy~$e$ on the pairs
      $(u, v')$ and $(u', v)$ of~$I'$, and copy~$e$
      \emph{except the fact $F_m$} on the pairs $(u, v)$ and $(u', v')$ of~$I'$.
    \end{itemize}
\end{defi}

\tikzset{midedge/.style={cyan,very thick}}

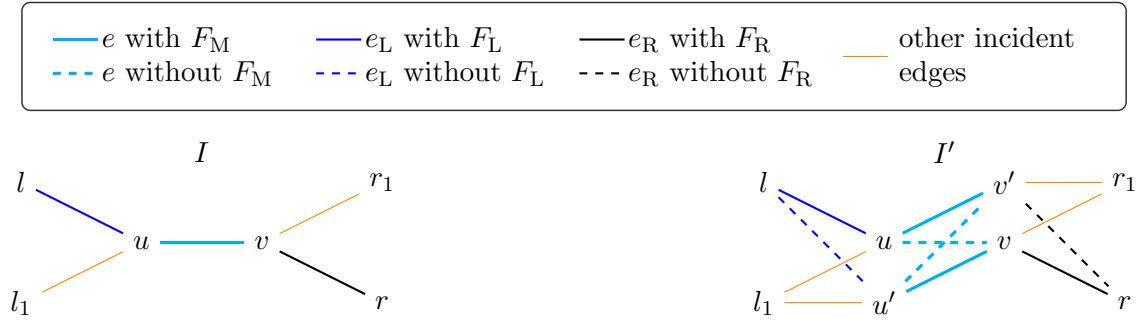
\begin{figure}
  \drawnewkeyb%
  \medskip

  \begin{tikzpicture}[yscale=.8,xscale=.8]
    \node (I) at (3, .5) {$I$};
    \node (vl) at (0, 0) {$l$};
    \node (vll) at (0, -2) {$l_1$};
    \node (u) at (2, -1) {$u$};
    \node (v) at (4, -1) {$v$};
    \node (ur) at (6, -2) {$r$};
    \node (urr) at (6, 0) {$r_1$};
    \draw (u) edge[very thick,midedge] (v);
    \draw (vl) edge[blue,thick] (u);
    \draw (vll) edge[orange] (u);
    \draw (v) edge[thick,black] (ur);
    \draw (v) edge[orange] (urr);
  \end{tikzpicture}
  \hfill
  \begin{tikzpicture}[yscale=.8,xscale=.8]
    \node (I) at (3, .5) {$I'$};
    \node (vl) at (0, 0) {$l$};
    \node (vll) at (0, -2) {$l_1$};
    \node (u) at (2, -1) {$u$};
    \node (up) at (2, -2) {$u'$};
    \node (v) at (4, -1) {$v$};
    \node (vp) at (4, 0) {$v'$};
    \node (ur) at (6, -2) {$r$};
    \node (urr) at (6, 0) {$r_1$};
    \draw (u) edge[cyan,very thick,midedge] (vp);
    \draw (up) edge[cyan,very thick,midedge] (v);
    \draw (u) edge[dashed,very thick,cyan,midedge] (v);
    \draw (up) edge[dashed,very thick,cyan,midedge] (vp);
    \draw (vl) edge[blue,thick] (u);
    \draw (vl) edge[blue,dashed,thick] (up);
    \draw (vll) edge[orange] (u);
    \draw (vll) edge[orange] (up);
    \draw (v) edge[black, thick] (ur);
    \draw (vp) edge[black,dashed,thick] (ur);
    \draw (v) edge[orange] (urr);
    \draw (vp) edge[orange] (urr);
  \end{tikzpicture}
  \caption{Example of fine dissociation from an instance~$I$ (left) to~$I'$ (middle) for a choice of~$e$, of~$\Pi = (F_{\ll}, F_{\rr})$, and of $F_{\mm}$. We call $e_{\ll}$ and~$e_{\rr}$ the edges of~$F_{\ll}$ and $F_{\rr}$.}%
  \label{fig:finedissociation}
\end{figure}

\noindent
Note that if the only non-unary fact covered by the edge~$e$ in~$I$ is~$F_{\mm}$, then $(u, v)$ and $(u', v')$ are not edges in the result of the fine dissociation; otherwise, they are edges but with a smaller weight than~$e$.
Observe that fine dissociation is related both to dissociation (Section~\ref{sec:findhard}) and to iteration
(Section~\ref{sec:pp2dnf}). We will study later when fine dissociation can make the query false.

We can now start the proof of Theorem~\ref{thm:ustcon} by describing the coding. It depends on our choice of~$I_{e,\Pi}$ and of a fact~$F_{\mm}$, but like in Section~\ref{sec:pp2dnf} it
does not depend on the query~$Q$. Given an
input st-graph~$G$, we construct a TID $\pdb$ whose possible worlds will have a bijection to those of~$G$. The process is illustrated on an example in Figure~\ref{fig:concode}, and defined formally below:
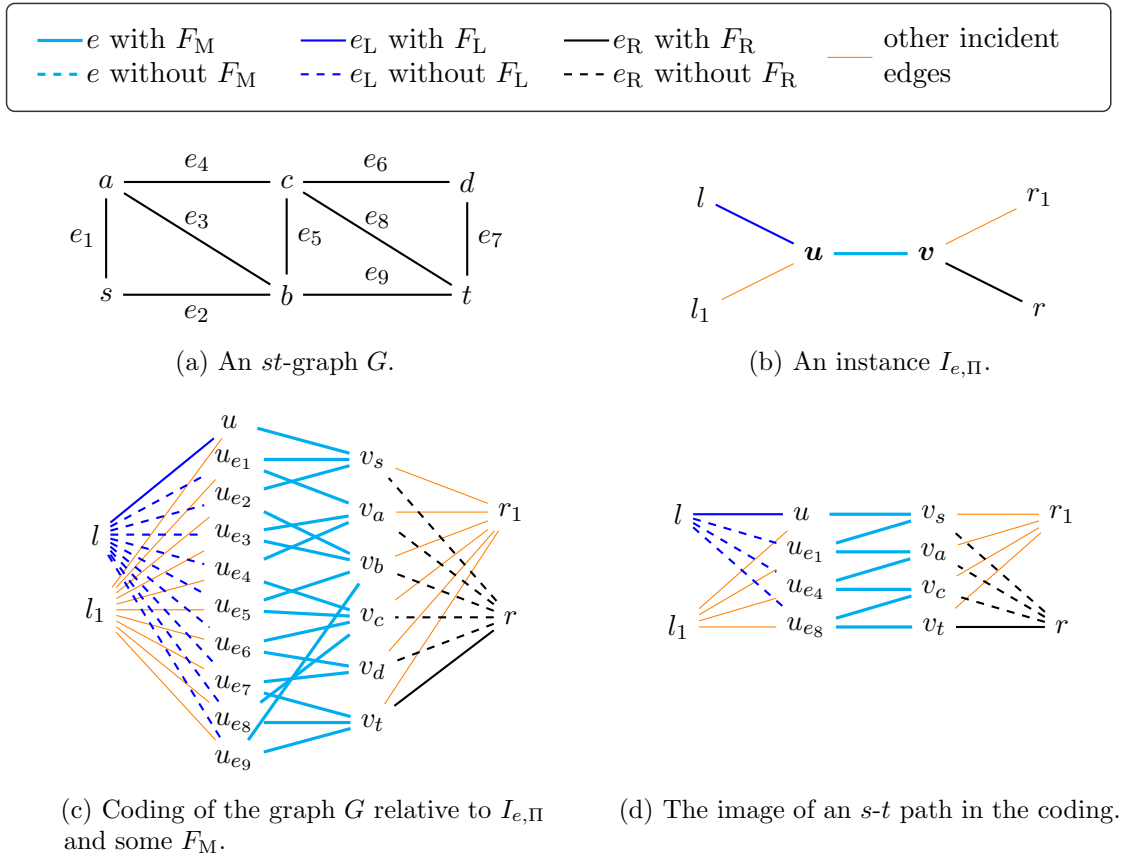
\begin{figure}[t]
	\begin{subfigure}[t]{\linewidth}
		\drawnewkeyb%
	\end{subfigure}

	\bigskip

	\begin{subfigure}[t]{.49\linewidth}
		\centering
		\begin{tikzpicture}[xscale=.8]
		\node (s) at (0, 0) {$s$};
		\node (a) at (0, 1.5) {$a$};
		\node (b) at (3, 0) {$b$};
		\node (c) at (3, 1.5) {$c$};
		\node (d) at (6, 1.5) {$d$};
		\node (t) at (6, 0) {$t$};

		\draw (s) edge[-, thick] node[midway, left]{$e_1$} (a);
		\draw (s) edge[-, thick] node[midway, below]{$e_2$} (b);
		\draw (a) edge[-, thick] node[midway, above]{$e_3$} (b);
		\draw (a) edge[-, thick] node[midway, above]{$e_4$} (c);
		\draw (b) edge[-, thick] node[midway, right]{$e_5$} (c);
		\draw (b) edge[-, thick] node[midway, above]{$e_9$} (t);
		\draw (c) edge[-, thick] node[midway, above]{$e_6$} (d);
		\draw (c) edge[-, thick] node[midway, above]{$e_8$} (t);
		\draw (d) edge[-, thick] node[midway, right]{$e_7$} (t);
		\end{tikzpicture}
		\caption{An $st$-graph $G$.}\label{fig:concode-a}
	\end{subfigure}
	\hfill
	\begin{subfigure}[t]{.49\linewidth}
		\centering
		\begin{tikzpicture}
		\node (v1) at (0, 1.5) {$l$};
		\node (v2) at (0, 0) {$l_1$};
		\node (u) at (1.5, 0.75) {$\bm{u}$};
		\node (v) at (3, 0.75) {$\bm{v}$};
		\node (u1) at (4.5, 1.5) {$r_1$};
		\node (u2) at (4.5, 0) {$r$};
		\draw (u) edge[-, very thick, cyan] (v);
		\draw (v1) edge[-,thick, blue]  (u);
		\draw (v2) edge[-,  orange] (u);
		\draw (v) edge[-,  orange] (u1);
		\draw (v) edge[-,thick, black] (u2);
		\end{tikzpicture}
		\caption{An instance $I_{e,\Pi}$.}\label{fig:concode-b}
	\end{subfigure}

	\bigskip

	\qquad\begin{subfigure}[t]{.42\linewidth}
		\centering
		\begin{tikzpicture}[xscale=.8]
		\node (v1) at (0, 0) {${l}$};
		\node (v2) at (0, -1) {${l_1}$};
		\node (u) at (2.3, 1.5) {$u~$};
		\node (ue1) at (2.3, 1) {$u_{e_1}$};
		\node (ue2) at (2.3, 0.5) {$u_{e_2}$};
		\node (ue3) at (2.3, 0) {$u_{e_3}$};
		\node (ue4) at (2.3, -0.5) {$u_{e_4}$};
		\node (ue5) at (2.3, -1) {$u_{e_5}$};
		\node (ue6) at (2.3, -1.5) {$u_{e_6}$};
		\node (ue7) at (2.3, -2) {$u_{e_7}$};
		\node (ue8) at (2.3, -2.5) {$u_{e_8}$};
		\node (ue9) at (2.3, -3) {$u_{e_9}$};
		\node (vs) at (4.6, 1) {$v_{s}$};
		\node (va) at (4.6, 0.3) {$v_{a}$};
		\node (vb) at (4.6, -0.4) {$v_{b}$};
		\node (vc) at (4.6, -1.1) {$v_{c}$};
		\node (vd) at (4.6, -1.8) {$v_{d}$};
		\node (vt) at (4.6, -2.5) {$v_{t}$};
		\node (u1) at (6.9, 0.3) {${r_1}$};
		\node (u2) at (6.9, -1.1) {${r}$};
		\draw (v2) edge[-, orange] (u);
		\draw (v2) edge[-, orange] (ue1);
		\draw (v2) edge[-, orange] (ue2);
		\draw (v2) edge[-, orange] (ue3);
		\draw (v2) edge[-, orange] (ue4);
		\draw (v2) edge[-, orange] (ue5);
		\draw (v2) edge[-, orange] (ue6);
		\draw (v2) edge[-, orange] (ue7);
		\draw (v2) edge[-, orange] (ue8);
		\draw (v2) edge[-, orange] (ue9);

		\draw (u1) edge[-, orange] (vs);
		\draw (u1) edge[-, orange] (va);
		\draw (u1) edge[-, orange] (vb);
		\draw (u1) edge[-, orange] (vc);
		\draw (u1) edge[-, orange] (vd);
		\draw (u1) edge[-, orange] (vt);
		\draw (vs) edge[-, very thick, cyan] (u);
		\draw (vs) edge[-, very thick,cyan] (ue1);
		\draw (vs) edge[-, very thick,cyan] (ue2);
		\draw (va) edge[-, very thick,cyan] (ue1);
		\draw (va) edge[-, very thick,cyan] (ue3);
		\draw (va) edge[-, very thick,cyan] (ue4);
		\draw (vb) edge[-, very thick,cyan] (ue2);
		\draw (vb) edge[-, very thick,cyan] (ue3);
		\draw (vb) edge[-, very thick,cyan] (ue5);
		\draw (vb) edge[-, very thick,cyan] (ue9);
		\draw (vc) edge[-, very thick,cyan] (ue4);
		\draw (vc) edge[-, very thick,cyan] (ue5);
		\draw (vc) edge[-, very thick,cyan] (ue6);
		\draw (vc) edge[-, very thick,cyan] (ue8);
		\draw (vd) edge[-, very thick,cyan] (ue6);
		\draw (vd) edge[-, very thick,cyan] (ue7);
		\draw (vt) edge[-, very thick,cyan] (ue7);
		\draw (vt) edge[-, very thick,cyan] (ue8);
		\draw (vt) edge[-, very thick,cyan] (ue9);
		\draw (u2) edge[-, thick,dashed, black] (vs);
		\draw (u2) edge[-, thick,dashed,black] (va);
		\draw (u2) edge[-, thick,dashed,black] (vb);
		\draw (u2) edge[-, thick,dashed,black] (vc);
		\draw (u2) edge[-, thick,dashed,black] (vd);
		\draw (u2) edge[-, thick,black] (vt);

		\draw (v1) edge[-, thick, blue] (u);
		\draw (v1) edge[-, thick, dashed,blue] (ue1);
		\draw (v1) edge[-, thick, dashed,blue] (ue2);
		\draw (v1) edge[-, thick, dashed,blue] (ue3);
		\draw (v1) edge[-, thick, dashed,blue] (ue4);
		\draw (v1) edge[-, thick, dashed,blue] (ue5);
		\draw (v1) edge[-, thick, dashed,blue] (ue6);
		\draw (v1) edge[-, thick, dashed,blue] (ue7);
		\draw (v1) edge[-, thick, dashed,blue] (ue8);
		\draw (v1) edge[-, thick, dashed,blue] (ue9);
		\end{tikzpicture}
		\caption{Coding of the graph~$G$ relative to~$I_{e,\Pi}$ and
			some~$F_{\mm}$.}\label{fig:concode-c}
	\end{subfigure}
	\hfill
	\begin{subfigure}[t]{.49\linewidth}
		\centering
		\begin{tikzpicture}
		\node (dumb) at (0, -2.9) {};
		\node (v1) at (0, 0.5) {${l}$};
		\node (v2) at (0, -1) {${l_1}$};
		\node (u) at (1.7, 0.5) {$u~$};
		\node (ue1) at (1.7, 0) {$u_{e_1}$};
		\node (ue4) at (1.7, -0.5) {$u_{e_4}$};
		\node (ue8) at (1.7, -1) {$u_{e_8}$};
		\node (vs) at (3.4, 0.5) {$v_{s}$};
		\node (va) at (3.4, 0) {$v_{a}$};
		\node (vc) at (3.4, -0.5) {$v_{c}$};
		\node (vt) at (3.4, -1) {$v_{t}$};
		\node (u1) at (5.1, 0.5) {${r_1}$};
		\node (u2) at (5.1, -1) {${r}$};
		\draw (v2) edge[-,  orange] (u);
		\draw (v2) edge[-, orange] (ue1);
		\draw (v2) edge[-, orange] (ue4);
		\draw (v2) edge[-, orange] (ue8);
		\draw (vs) edge[-, very thick, cyan] (u);
		\draw (vs) edge[-, very thick,cyan] (ue1);
		\draw (va) edge[-, very thick,cyan] (ue1);
		\draw (va) edge[-, very thick,cyan] (ue4);
		\draw (vc) edge[-, very thick,cyan] (ue4);
		\draw (vc) edge[-, very thick,cyan] (ue8);
		\draw (vt) edge[-, very thick,cyan] (ue8);
		\draw (u1) edge[-,  orange] (vs);
		\draw (u1) edge[-, orange] (va);
		\draw (u1) edge[-, orange] (vc);
		\draw (u1) edge[-, orange] (vt);
		\draw (u2) edge[-, thick,dashed, black] (vs);
		\draw (u2) edge[-, thick,dashed,black] (va);
		\draw (u2) edge[-, thick,dashed,black] (vc);
		\draw (u2) edge[-, thick,black] (vt);

		\draw (v1) edge[-, thick, blue] (u);
		\draw (v1) edge[-, thick, dashed,blue] (ue1);
		\draw (v1) edge[-, thick, dashed,blue] (ue4);
		\draw (v1) edge[-, thick, dashed,blue] (ue8);
		\end{tikzpicture}
		\caption{The image of an $s$-$t$ path in the coding.}\label{fig:concode-d}
	\end{subfigure}
	\caption{Example of the coding on an $st$-graph $G$ shown in Figure~\ref{fig:concode-a}.
		We encode $G$ relative to an instance $I_{e,\Pi}$ shown in Figure~\ref{fig:concode-b}, and relative to some choice of a non-unary fact $F_{\mm}$ covered by~$e$.
		The coding of~$G$ relative to~$I_{e,\Pi}$ and~$F_{\mm}$ is shown in Figure~\ref{fig:concode-c}, with the probabilistic facts being \emph{exactly one} copy of~$F_{\mm}$ for \emph{one} of every pair of cyan edges adjacent to an element in $\{u_{e_1}, \ldots, u_{e_9}\}$.
		Each $st$-path in $G$ gives rise to a subinstance in the coding: consider for instance the $st$-path which is via the edges $e_1, e_4,e_8$. The corresponding subinstance in the coding for this path is shown in Figure~\ref{fig:concode-d}: it is an iterate of the form~$I_{e,\Pi}^{n+1}$ where $n$ is the number of edges on the path (here $n=3$).
	}%
	\label{fig:concode}
\end{figure}

\begin{defi}%
  \label{def:concoding}
  Let $I_{e,\Pi}$ be a $\sigma$-instance where $e = (u, v)$, $\Pi = (F_{\ll}, F_{\rr})$, $F_{\ll} = R_{\ll}(l, u)$, $F_{\rr} = R_{\rr}(v, r)$ and let $F_{\mm}$ be a non-unary fact of~$I$ covered by~$e$. Let $G = (W, C)$ be an st-graph with source~$s$ and target~$t$.
  The \emph{coding} of~$G$ relative to~$I_{e,\Pi}$ and~$F_{\mm}$ is a TID $\pdb = (J, \pi)$ with domain $\dom(J) \colonequals \dom(I) \cup \{u_c \mid c \in C\} \cup \{v_w \mid w \in W \setminus \{t\}\}$, where the new elements are fresh, and where we use~$v_t$ to refer to~$v$ for convenience.
  The facts of the $\sigma$-instance~$J$ and the probability mapping~$\pi$ are defined as follows:
  \begin{itemize}
    \item \emph{Copy non-incident facts:}
    Initialize $J$ as the induced subinstance of~$I$ on $\dom(I) \setminus \{u, v\}$.
    \item \emph{Copy incident facts $F_{\ll}$ and $F_{\rr}$:}
    Add the facts $F_{\ll}$ and $F_{\rr}$ to~$J$.
    \item \emph{Copy other left-incident facts:}
    For every $\sigma^\leftrightarrow$-fact $F_{\ll}' = R_{\ll}'(l', u)$ of~$I$ that is left-incident to~$e$ (i.e., $l' \notin \{u, v\}$) and where $F_{\ll}' \neq F_{\ll}$, add to~$J$ the facts $R_{\ll}'(l', u_c)$ for each edge $c \in C$.
    \item \emph{Copy other right-incident facts:}
     For every $\sigma^\leftrightarrow$-fact $F_{\rr}' = R_{\rr}'(v, r')$ of~$I$ that is right-incident to~$e$ (i.e., $r' \notin \{u, v\}$) and where $F_{\rr}' \neq F_{\rr}$, add to~$J$ the facts $R_{\rr}'(v_w, r')$ for each $w \in W$.
    \item \emph{Create copies of~$e$:}
    Copy~$e$ on the pair $(u,v_s)$ of~$J$, and for each edge $c = \{a, b\}$ in~$C$, copy~$e$ on the pairs $(u_c, v_a)$ and $(u_c, v_b)$ of~$J$.
    \end{itemize}
  Finally, we define the function $\pi$ as follows. For each edge $c$ of~$C$, we choose one arbitrary vertex $w \in c$, and set $\pi$ to map the copy of the fact~$F_{\mm}$ in the edge $(u_c, v_w)$ to~$0.5$, All other facts are mapped to~$1$ by~$\pi$.
\end{defi}

It is important to note that the edges are coded by paths of length~2. This choice is critical, because the source graph in the reduction is undirected, but the facts on edges are directed; so, intuitively, we symmetrize by having two copies of the edge in opposite directions in order to traverse them in both ways.
The choice on how to orient the edges (i.e., the choice of $w \in c$ when defining~$\pi$) has no impact in how the edges can be traversed when their probabilistic fact is present, but it has an impact when the probabilistic fact is missing. Indeed, this is the reason why fine dissociation includes two copies of~$e$ with one missing fact.

It is easy to see that the given coding is in polynomial time in the input $G$ for every choice of~$I_{e,\Pi}$ and~$F_{\mm}$. Let us now define the bijection~$\phi$, mapping each possible world~$\omega$ of~$G$ to a possible world of the TID $\pdb$ as follows. For each edge $c \in C$, we keep the probabilistic fact incident to~$u_c$ in the instance $\phi(\omega)$  if $c$ is kept in the possible world~$\omega$, and we do not keep it otherwise. It is obvious that this correspondence is bijective and that all possible worlds have the same probability $0.5^{\card{C}}$. We can now explain why~$\phi$ defines a reduction.
Recall from Definition~\ref{def:stcon} that a possible world of~$G$ is \emph{good} if it contains an $s,t$-path, and \emph{bad} otherwise. Here is the formal statement:
\begin{prop}%
	\label{prp:concoding}
        Let the TID $\pdb = (J, \pi)$ be the coding of an undirected st-graph $G$ relative to an instance $I_{e,\Pi}$ and to~$F_{\mm}$ as described in Definition~\ref{def:concoding}. Let $\phi$ be the bijective function defined above from the possible worlds of $G$ to those of $\pdb$. Then:
	\begin{enumerate}
          \item For any \emph{good} possible world $\omega$ of~$G$ with a witnessing simple $s,t$-path traversing $n$ edges, $\phi(\omega)$ has a homomorphism \emph{from} the iterate~$I_{e,\Pi}^{n+1}$.
          \item For any \emph{bad} possible world $\omega$ of~$G$, $\phi(\omega)$ has a homomorphism \emph{to} the result of finely dissociating~$e$ in~$I$ relative to~$\Pi$ and~$F_{\mm}$.
	\end{enumerate}
\end{prop}
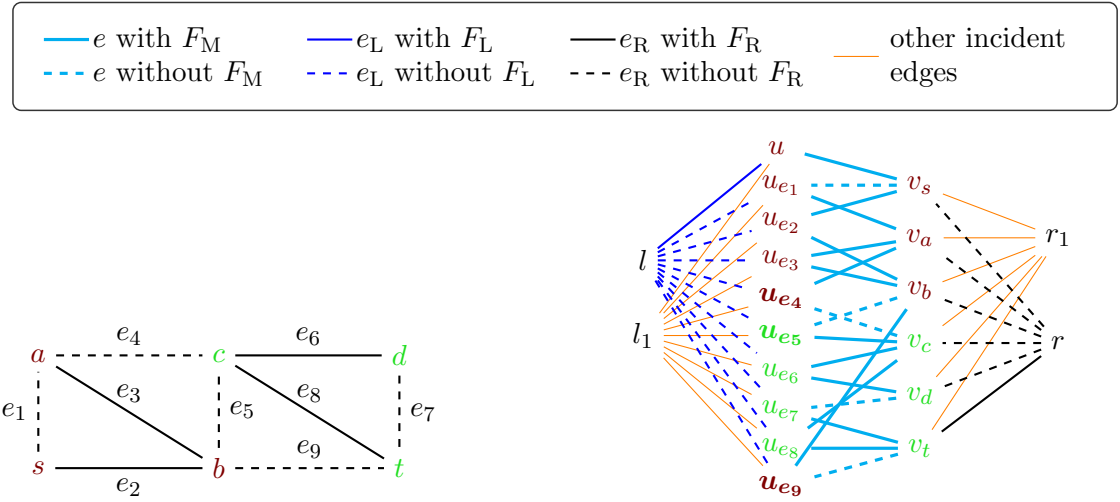
\begin{figure}
	\drawnewkeyb%
	\medskip

	\begin{subfigure}[t]{2.4in}
		\noindent\begin{tikzpicture}[xscale=.8]
		\node[darkred] (s) at (0, 0) {$s$};
		\node[darkred] (a) at (0, 1.5) {$a$};
		\node[darkred] (b) at (3, 0) {$b$};
		\node[darkgreen] (c) at (3, 1.5) {$c$};
		\node[darkgreen] (d) at (6, 1.5) {$d$};
		\node[darkgreen] (t) at (6, 0) {$t$};

		\draw (s) edge[-, thick, dashed] node[midway, left]{$e_1$} (a);
		\draw (s) edge[-, thick] node[midway, below]{$e_2$} (b);
		\draw (a) edge[-, thick] node[midway, above]{$e_3$} (b);
		\draw (a) edge[-, thick, dashed] node[midway, above]{$e_4$} (c);
		\draw (b) edge[-, thick, dashed] node[midway, right]{$e_5$} (c);
		\draw (b) edge[-, thick, dashed] node[midway, above]{$e_9$} (t);
		\draw (c) edge[-, thick] node[midway, above]{$e_6$} (d);
		\draw (c) edge[-, thick] node[midway, above]{$e_8$} (t);
		\draw (d) edge[-, thick, dashed] node[midway, right]{$e_7$} (t);
		\end{tikzpicture}
		\caption{A possible world $\omega$ of~$G$ with no $s,t$-path
			(dashed edges are the ones that are not kept): the vertices are colored in
			\textcolor{darkred}{red} or \textcolor{darkgreen}{green} depending on their side of the cut.}\label{fig:concode-e}
	\end{subfigure}
	\hfill
	\begin{subfigure}[t]{3in}
		\centering
		\begin{tikzpicture}[xscale=.8]
		\node (v1) at (0, 0) {${l}$};
		\node (v2) at (0, -1) {${l_1}$};
		\node[darkred] (u) at (2.3, 1.5) {$u~$};
		\node[darkred] (ue1) at (2.3, 1) {$u_{e_1}$};
		\node[darkred] (ue2) at (2.3, 0.5) {$u_{e_2}$};
		\node[darkred] (ue3) at (2.3, 0) {$u_{e_3}$};
		\node[darkred] (ue4) at (2.3, -0.5) {$\bm{u_{e_4}}$};
		\node[darkgreen] (ue5) at (2.3, -1) {$\bm{u_{e_5}}$};
		\node[darkgreen] (ue6) at (2.3, -1.5) {$u_{e_6}$};
		\node[darkgreen] (ue7) at (2.3, -2) {$u_{e_7}$};
		\node[darkgreen] (ue8) at (2.3, -2.5) {$u_{e_8}$};
		\node[darkred] (ue9) at (2.3, -3) {$\bm{u_{e_9}}$};
		\node[darkred] (vs) at (4.6, 1) {$v_{s}$};
		\node[darkred] (va) at (4.6, 0.3) {$v_{a}$};
		\node[darkred] (vb) at (4.6, -0.4) {$v_{b}$};
		\node[darkgreen] (vc) at (4.6, -1.1) {$v_{c}$};
		\node[darkgreen] (vd) at (4.6, -1.8) {$v_{d}$};
		\node[darkgreen] (vt) at (4.6, -2.5) {$v_{t}$};
		\node (u1) at (6.9, 0.3) {${r_1}$};
		\node (u2) at (6.9, -1.1) {${r}$};
		\draw (v2) edge[-, orange] (u);
		\draw (v2) edge[-, orange] (ue1);
		\draw (v2) edge[-, orange] (ue2);
		\draw (v2) edge[-, orange] (ue3);
		\draw (v2) edge[-, orange] (ue4);
		\draw (v2) edge[-, orange] (ue5);
		\draw (v2) edge[-, orange] (ue6);
		\draw (v2) edge[-, orange] (ue7);
		\draw (v2) edge[-, orange] (ue8);
		\draw (v2) edge[-, orange] (ue9);
		\draw (vs) edge[-, very thick, cyan] (u);
		\draw (vs) edge[-, very thick,cyan,dashed] (ue1);
		\draw (vs) edge[-, very thick,cyan] (ue2);
		\draw (va) edge[-, very thick,cyan] (ue1);
		\draw (va) edge[-, very thick,cyan] (ue3);
		\draw (va) edge[-, very thick,cyan] (ue4);
		\draw (vb) edge[-, very thick,cyan] (ue2);
		\draw (vb) edge[-, very thick,cyan] (ue3);
		\draw (vb) edge[-, very thick,cyan,dashed] (ue5);
		\draw (vb) edge[-, very thick,cyan] (ue9);
		\draw (vc) edge[-, very thick,cyan,dashed] (ue4);
		\draw (vc) edge[-, very thick,cyan] (ue5);
		\draw (vc) edge[-, very thick,cyan] (ue6);
		\draw (vc) edge[-, very thick,cyan] (ue8);
		\draw (vd) edge[-, very thick,cyan] (ue6);
		\draw (vd) edge[-, very thick,cyan,dashed] (ue7);
		\draw (vt) edge[-, very thick,cyan] (ue7);
		\draw (vt) edge[-, very thick,cyan] (ue8);
		\draw (vt) edge[-, very thick,cyan,dashed] (ue9);
		\draw (u1) edge[-, orange] (vs);
		\draw (u1) edge[-, orange] (va);
		\draw (u1) edge[-, orange] (vb);
		\draw (u1) edge[-, orange] (vc);
		\draw (u1) edge[-, orange] (vd);
		\draw (u1) edge[-, orange] (vt);
		\draw (v1) edge[-, thick, blue] (u);
		\draw (v1) edge[-, thick, dashed,blue] (ue1);
		\draw (v1) edge[-, thick, dashed,blue] (ue2);
		\draw (v1) edge[-, thick, dashed,blue] (ue3);
		\draw (v1) edge[-, thick, dashed,blue] (ue4);
		\draw (v1) edge[-, thick, dashed,blue] (ue5);
		\draw (v1) edge[-, thick, dashed,blue] (ue6);
		\draw (v1) edge[-, thick, dashed,blue] (ue7);
		\draw (v1) edge[-, thick, dashed,blue] (ue8);
		\draw (v1) edge[-, thick, dashed,blue] (ue9);
		\draw (u2) edge[-, thick,dashed, black] (vs);
		\draw (u2) edge[-, thick,dashed,black] (va);
		\draw (u2) edge[-, thick,dashed,black] (vb);
		\draw (u2) edge[-, thick,dashed,black] (vc);
		\draw (u2) edge[-, thick,dashed,black] (vd);
		\draw (u2) edge[-, thick,black] (vt);
		\end{tikzpicture}
		\caption{Possible world of the coding in Figure~\ref{fig:concode-c} for the possible world
			of~$G$ at the left. Copies of~$e$ are dashed when they are missing the fact~$F_{\mm}$.
			Vertices $u_{e_i}$ corresponding to edges across the cut are in bold.
		}\label{fig:concode-f}
	\end{subfigure}
	\caption{Illustration of a possible world (Figure~\ref{fig:concode-e}) of the graph $G$ from Figure~\ref{fig:concode-a}, and the corresponding possible world (Figure~\ref{fig:concode-f}) of the coding (Figure~\ref{fig:concode-c}). The homomorphism of Figure~\ref{fig:concode-f} to the fine dissociation is given by the vertex colors:
		the red $u$-vertices are mapped to~$u$, the red $v$-vertices are mapped to~$v'$, the green $u$-vertices are mapped to~$u'$, and the green $v$-vertices are mapped to~$v$. The vertex colors are determined by the cut (Figure~\ref{fig:concode-e}) except for the bold vertices where it depends on the orientation choice.}%
	\label{fig:concode2}
\end{figure}
\begin{proof}
  As before, we start with the easier forward direction (1), and then prove the backward direction (2).

\begin{enumerate}
    \item 	Consider a witnessing path $s = w_1, \ldots, w_{n+1} = t$ in the possible world $\omega$ of~$G$, and assume without loss of generality that the path is simple, i.e., it traverses each vertex at most once. We claim that the possible world $J' \colonequals \phi(\omega)$ actually has a subinstance isomorphic to~$I_{e,\Pi}^{n+1}$. See Figure~\ref{fig:concode-d} for an example.

  To see why this is true, we take as usual the facts of~$J'$ that do not involve any copy of~$u$ or~$v$ and keep them as-is, because they occur in~$J'$ as they do in~$I_{e,\Pi}^{n+1}$.
  Now, we start by taking the one copy of~$F_{\ll}$ leading to~$u$ and the copy of~$e$ leading to~$v_s$. We now follow the path which gives a path of copies of~$e$: for each edge $c =\{w_j,w_{j+1}\}$ of the path, we have two successive copies of~$e$ between~$v_{w_j}$ and~$u_c$, and between $u_c$ and~$v_{w_{j+1}}$.
  Note that, as the path uses edge~$c$, it was kept in~$\omega$, so all the copies of~$e$ in question have all their facts, i.e., neither of the copies of~$F_{\mm}$ can be missing. The assumption that the path is simple ensures that we do not visit the same vertex multiple times. After traversing these $2n$ copies of~$e$ in alternating directions, we reach $v_t = v$, and finally we use the fact $F_{\rr}$ which is incident to~$v$. So, we have indeed found a subinstance of~$J'$ which is isomorphic to~$I_{e,\Pi}^{n+1}$.

\item	Let us write $J' \colonequals \phi(\omega)$, and let us write $e = (u, v)$.
  Let us denote by $I'$ the result of finely dissociating in~$I$ the edge~$e$ relative to the incident pair~$\Pi$ and the fact~$F_{\mm}$: this is depicted in Figure~\ref{fig:finedissociation}. Let us show that $J'$ has a homomorphism to~$I'$. See Figure~\ref{fig:concode-f} for an example of a bad possible world~$J'$, and Figure~\ref{fig:concode-e} for the corresponding possible world~$\omega$ of~$G$.

  We use the fact that, as the possible world~$\omega$ of~$G$ has no path from~$s$ to~$t$, there is an \emph{$s,t$-cut of~$\omega$}, i.e., a function $\psi$ mapping each vertex of~$G$ to either $\LL$ or~$\RR$ such that $s$ is mapped to~$\LL$, $t$ is mapped to~$\RR$, and every edge $\{x, y\}$ for which $\psi(x) \neq \psi(y)$ was not kept in~$\omega$. See Figure~\ref{fig:concode-e} for an illustration, where the red vertices are mapped to~$\LL$ and the green vertices are mapped to~$\RR$.

  We map~$u$ in~$J'$ to~$u$ in~$I'$ and $v_s$ to~$v$, which maps the copy of~$e$ between~$u$ and~$v_s$ in~$J'$ to a copy of~$e$ in~$I'$. Now observe that we can map to~$v'$ in~$I'$ all the nodes~$v_w$ such that $\psi(w) = \LL$, including $v_s$. The edges between these nodes in~$J'$, whether they were kept in~$\omega$ or not, are mapped by going back-and-forth on the edge $(u, v')$ in~$I'$.
  In Figure~\ref{fig:concode-f}, this defines the image of $v_a$, $v_b$, and $u_{e_1}, u_{e_2}, u_{e_3}$ corresponding to the edges between them. In the same way we can map to~$v$ in~$I'$ all the nodes~$v_w$ such that $\psi(w) = \RR$, including $v_t$ and all edges between these nodes, going back-and-forth on edge $(u', v)$ in~$I'$. In Figure~\ref{fig:concode-f}, this defines the image of $v_c$, $v_d$, $v_t$, and $u_{e_6}, u_{e_7}, u_{e_8}$ corresponding to the edges between them.

  We must still map the edges across the cut, i.e., edges $c = \{x, y\}$ such that $\psi(x) = \LL$ and $\psi(y) = \RR$. In~$J'$, these edges give rise to two edges $(u_c, v_x)$ and $(u_c, v_y)$, one of which is a copy of~$e$ and the other one is a copy of~$e$ with the fact~$F_{\mm}$ missing --- which one is which depends on the arbitrary orientation choice that we made when defining~$\pi$. Depending on the case, we map~$u_c$ either to~$u$ or to~$u'$ so that the two incident edges to~$u_c$ are mapped in~$I'$ either to~$(u,v')$ (a copy of~$e$) and $(u, v)$ (a copy of~$e$ minus~$F_{\mm}$), or to $(u',v')$ (a copy of~$e$ minus~$F_{\mm}$) and $(u', v)$ (a copy of~$e$).
  In Figure~\ref{fig:concode-f}, this allows us to define the image of the bold vertices ($u_{e_4}, u_{e_5}, u_{e_9}$) corresponding to the edges across the cut. We follow the orientation choice when defining~$\pi$, which can be seen by examining which edges are dashed, and we map $u_{e_4}$ and $u_{e_9}$ to~$u$ and map $u_{e_5}$ to~$v$.

  Thus, we have explained how we map the copies of~$u$ and~$v$, the copies of~$e$ (including the ones without~$F_{\mm}$), and the two facts $F_{\ll}$ and~$F_{\rr}$. As usual, we have not discussed the facts that do not involve a copy of~$u$ or~$v$ in~$J'$, or the facts that involve one of them and are not facts of~$e$, $F_{\ll}$, or~$F_{\rr}$, but these are found in~$I'$ in the same way that they occur in~$J'$ (noting that we have only mapped copies of~$u$ to copies of~$u$, and copies of~$v$ to copies of~$v$). This concludes the definition of the homomorphism and concludes the proof. \qedhere
  \end{enumerate}
\end{proof}

\noindent
Proposition~\ref{prp:concoding} leads us to a proof of Theorem~\ref{thm:ustcon}: good possible worlds of~$G$ give a possible world of~$\pdb$ that satisfies~$Q$ thanks to the iterability of~$e$, and bad possible
worlds of~$G$ give a possible world of~$\pdb$ having a homomorphism to the fine dissociation. The only missing piece is to argue that the fine dissociation does not satisfy the query. We can do this using the minimality and tightness of the pattern:

\begin{lem}%
	\label{lem:magic}
  Let $Q$ be a query, let $(I, e)$ be a minimal tight pattern for~$Q$, let $\Pi$ be an arbitrary incident pair of~$e$ in~$I$, and let~$F_{\mm}$ be an arbitrary non-unary fact covered by~$e$ in~$I$.
  Then, the result of the fine dissociation of~$e$ in~$I$ relative to $\Pi$ and $F_{\mm}$ does not satisfy~$Q$.
\end{lem}
\begin{proof}
  We assume that the fine dissociation~$I_1$ satisfies~$Q$, and show a contradiction by rewriting it in several steps. The process of the proof is illustrated as Figure~\ref{fig:magic}.

    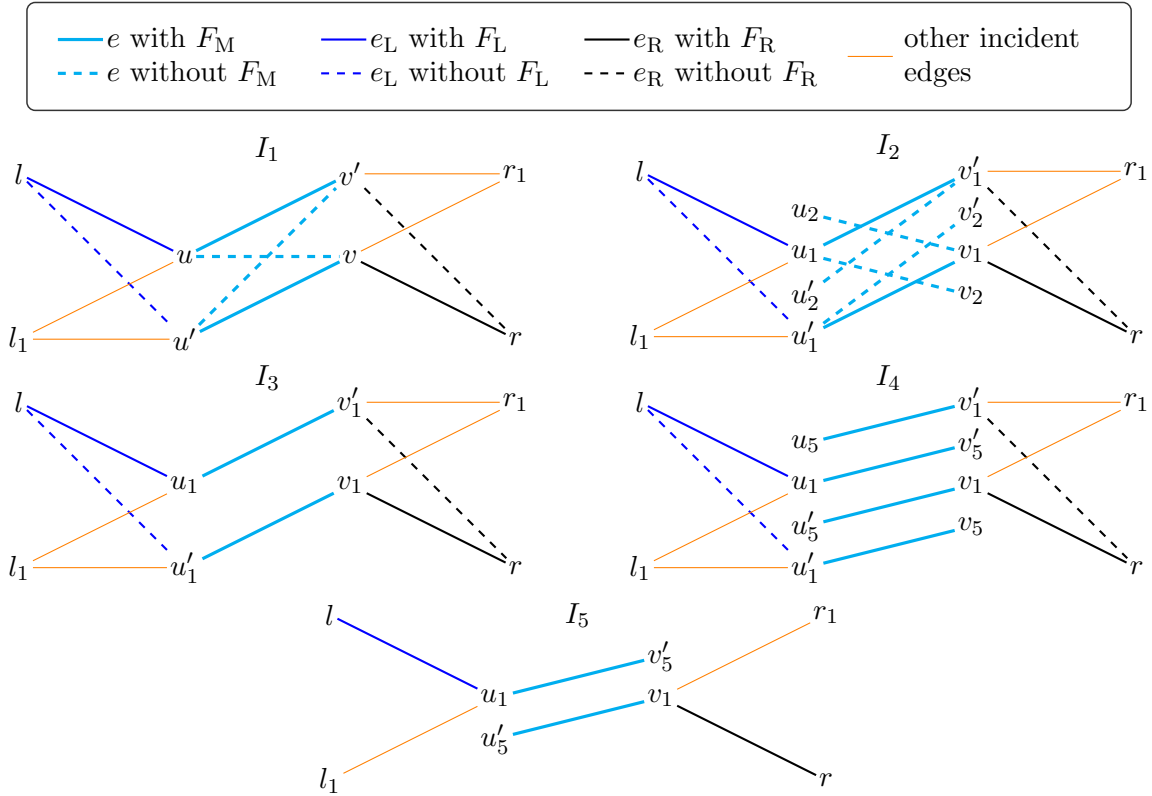
\begin{figure}
  \drawnewkeyb%
  \medskip

	\centering
	\begin{tikzpicture}[scale=1.1, inner sep=1pt]
	\node (I) at (3, .3) {$I_1$};
	\node (vl) at (0, 0) {$l$};
	\node (vll) at (0, -2) {$l_1$};
	\node (u) at (2, -1) {$u$};
	\node (up) at (2, -2) {$u'$};
	\node (v) at (4, -1) {$v$};
	\node (vp) at (4, 0) {$v'$};
	\node (ur) at (6, -2) {$r$};
	\node (urr) at (6, 0) {$r_1$};
	\draw (u) edge[very thick,cyan] (vp);
	\draw (up) edge[very thick,cyan] (v);
	\draw (u) edge[dashed,very thick,cyan] (v);
	\draw (up) edge[dashed,very thick,cyan] (vp);
	\draw (vl) edge[thick,blue] (u);
	\draw (vl) edge[dashed,thick,blue] (up);
	\draw (vll) edge[orange] (u);
	\draw (vll) edge[orange] (up);
	\draw (v) edge[thick,black] (ur);
	\draw (vp) edge[dashed,thick] (ur);
	\draw (v) edge[orange] (urr);
	\draw (vp) edge[orange] (urr);
	\end{tikzpicture}
  \hfill
	\begin{tikzpicture}[scale=1.1, inner sep=1pt]
	\node (I) at (3, .3) {$I_2$};
	\node (vl) at (0, 0) {$l$};
	\node (vll) at (0, -2) {$l_1$};
	\node (u) at (2, -1) {$u_1$};
	\node (up) at (2, -2) {$u_1'$};
	\node (v) at (4, -1) {$v_1$};
	\node (vp) at (4, 0) {$v_1'$};
	\node (u2) at (2, -.5) {$u_2$};
	\node (up2) at (2, -1.5) {$u_2'$};
	\node (v2) at (4, -1.5) {$v_2$};
	\node (vp2) at (4, -.5) {$v_2'$};
	\node (ur) at (6, -2) {$r$};
	\node (urr) at (6, 0) {$r_1$};
	\draw (u) edge[very thick,cyan] (vp);
	\draw (up) edge[very thick,cyan] (v);
	\draw (u) edge[dashed,very thick,cyan] (v2);
	\draw (up) edge[dashed,very thick,cyan] (vp2);
	\draw (u2) edge[dashed,very thick,cyan] (v);
	\draw (up2) edge[dashed,very thick,cyan] (vp);
	\draw (vl) edge[thick,blue] (u);
	\draw (vl) edge[dashed,thick,blue] (up);
	\draw (vll) edge[orange] (u);
	\draw (vll) edge[orange] (up);
	\draw (v) edge[thick,black] (ur);
	\draw (vp) edge[dashed,thick] (ur);
	\draw (v) edge[orange] (urr);
	\draw (vp) edge[orange] (urr);
	\end{tikzpicture}

        \smallskip

	\begin{tikzpicture}[scale=1.1, inner sep=1pt]
	\node (I) at (3, .3) {$I_3$};
	\node (vl) at (0, 0) {$l$};
	\node (vll) at (0, -2) {$l_1$};
	\node (u) at (2, -1) {$u_1$};
	\node (up) at (2, -2) {$u_1'$};
	\node (v) at (4, -1) {$v_1$};
	\node (vp) at (4, 0) {$v_1'$};
	\node (ur) at (6, -2) {$r$};
	\node (urr) at (6, 0) {$r_1$};
	\draw (u) edge[very thick,cyan] (vp);
	\draw (up) edge[very thick,cyan] (v);
	\draw (vl) edge[thick,blue] (u);
	\draw (vl) edge[dashed,thick,blue] (up);
	\draw (vll) edge[orange] (u);
	\draw (vll) edge[orange] (up);
	\draw (v) edge[thick,black] (ur);
	\draw (vp) edge[dashed,thick] (ur);
	\draw (v) edge[orange] (urr);
	\draw (vp) edge[orange] (urr);
	\end{tikzpicture}
     \hfill
	\begin{tikzpicture}[scale=1.1, inner sep=1pt]
	\node (I) at (3, .3) {$I_4$};
	\node (vl) at (0, 0) {$l$};
	\node (vll) at (0, -2) {$l_1$};
	\node (u) at (2, -1) {$u_1$};
	\node (up) at (2, -2) {$u_1'$};
	\node (v) at (4, -1) {$v_1$};
	\node (vp) at (4, 0) {$v_1'$};
	\node (u2) at (2, -.5) {$u_5$};
	\node (up2) at (2, -1.5) {$u_5'$};
	\node (v2) at (4, -1.5) {$v_5$};
	\node (vp2) at (4, -.5) {$v_5'$};
	\node (ur) at (6, -2) {$r$};
	\node (urr) at (6, 0) {$r_1$};
	\draw (u2) edge[very thick,cyan] (vp);
	\draw (up2) edge[very thick,cyan] (v);
	\draw (u) edge[very thick,cyan] (vp2);
	\draw (up) edge[very thick,cyan] (v2);
	\draw (vl) edge[thick,blue] (u);
	\draw (vl) edge[dashed,thick,blue] (up);
	\draw (vll) edge[orange] (u);
	\draw (vll) edge[orange] (up);
	\draw (v) edge[thick,black] (ur);
	\draw (vp) edge[dashed,thick] (ur);
	\draw (v) edge[orange] (urr);
	\draw (vp) edge[orange] (urr);
	\end{tikzpicture}

    \medskip

	\begin{tikzpicture}[scale=1.1, inner sep=1pt]
	\node (I) at (3, 0) {$I_5$};
	\node (vl) at (0, 0) {$l$};
	\node (vll) at (0, -2) {$l_1$};
	\node (u) at (2, -1) {$u_1$};
	\node (v) at (4, -1) {$v_1$};
	\node (up2) at (2, -1.5) {$u_5'$};
	\node (vp2) at (4, -.5) {$v_5'$};
	\node (ur) at (6, -2) {$r$};
	\node (urr) at (6, 0) {$r_1$};
	\draw (up2) edge[very thick,cyan] (v);
	\draw (u) edge[very thick,cyan] (vp2);
	\draw (vl) edge[thick,blue] (u);
	\draw (vll) edge[orange] (u);
	\draw (v) edge[thick,black] (ur);
	\draw (v) edge[orange] (urr);
	\end{tikzpicture}
	\caption{Illustration of the proof of Lemma~\ref{lem:magic}, with $I_1$ being the fine dissociation $I'$
		of Figure~\ref{fig:finedissociation}, and $I_5$ being isomorphic to the dissociation on
		Figure~\ref{fig:dissociation}.}%
	\label{fig:magic}
\end{figure}

  Fix the query~$Q$, the minimal tight pattern $(I, e)$, and the choice of~$F_{\ll}$, $F_{\mm}$, and~$F_{\rr}$. Assume by way of contradiction that the result $I_1$ of the fine dissociation satisfies the query~$Q$. Consider now the edges $e_1'' = (u, v)$ and $e_1' = (u', v')$: their weight in~$I_1$, by construction, is one less than the weight of~$e$. Hence, as $(I, e)$ is minimal, by Definition~\ref{def:minimal}, we know that each of these edges cannot be tight: if one of these edges were, say~$e_1'$, then $(I_1, e_1')$ would be a tight pattern with $e_1'$ having a strictly smaller weight, which is impossible. Thus, as we assumed that~$I$ satisfies~$Q$, it must mean that we can dissociate~$e_1'$, then~$e_1''$ using the dissociation process of Definition~\ref{def:coarsedissociation} without violating~$Q$.
  Formally, we first dissociate $e_1'' = (u, v)$ to remove this edge, rename $u$ and $v$ to $u_1$ and $v_1$, create $u_2$ and $v_2$, and add back copies of the edge from $u_1$ to $v_2$ and from $u_2$ to~$v_1$. The dissociated edge is not tight as we argued, so $Q$ is still satisfied in the result~$I_1'$. Second, we dissociate $e_1' = (u', v')$, remove $e_1'$, rename $u'$ and $v'$ to $u_1'$ and $v_1'$, create $u_2'$ and $v_2'$, and create copies of~$e_1'$ from $u_1'$ to~$v_2'$ and from $u_2'$ to~$v_1'$. The dissociated edge $e_1'$ has the same weight in~$I_1'$ as it did in~$I_1$, so again it is not tight, and $Q$ still holds in the result~$I_2$. (See Figure~\ref{fig:magic}.)

  Note that $u_2$, $v_2$, $u_2'$, $v_2'$ are leaf vertices in~$I_2$, which only occur on the copies of the dissociated edges (the edges with the same facts as $e$ except~$F_{\mm}$). We have copies of the edge $e$ (from the fine dissociation) from $u_1$ to~$v_1'$ and from~$u_1'$ to~$v_1$.

  Observe now that we can map the leaves $u_2$, $v_2$, $u_2'$ and $v_2'$ to define a homomorphism:
	\begin{itemize}
          \item we map $u_2$ to~$u_1'$ and map the edge $(u_2, v_1)$ to the
            edge $(u_1', v_1)$ whose facts are those of~$e$, so a superset of the facts;
          \item we map $v_2$ to~$v_1'$ and map the edge $(u_1, v_2)$ to $(u_1, v_1')$;
          \item we map $u_2'$ to~$u_1$ and map the edge $(u_2', v_1')$ to the edge
            $(u_1, v_1')$;
          \item we map $v_2'$ to~$v_1$ and map the edge $(u_1', v_2')$ to the edge
            $(u_1', v_1)$.
	\end{itemize}

\noindent
  The resulting instance $I_3$ (see Figure~\ref{fig:magic}) is a homomorphic image of~$I_2$, so it still satisfies~$Q$. Relative to~$I_1$, it is the result of replacing $u$ with copies $u_1, u_1'$, and $v$ with copies $v_1, v_1'$, and having one copy of~$e$ from~$u_1'$ to~$v_1$ and from~$u_1$ to~$v_1'$, with all facts incident to~$u$ and~$v$ replicated on $u_1, u_1'$ and $v_1, v_1'$, except $F_{\ll}$ and $F_{\rr}$ which only involve $u_1$ and~$v_1$.
  In other words, the instance $I_3$ is isomorphic to the result~$I_1$ of the fine dissociation (Figure~\ref{fig:finedissociation}), except that we have not created copies of~$e$ without~$F_{\mm}$ between $u_1$ and~$v_1$ and between $u_1'$ and~$v_1'$. We have justified, from our assumption that $I_1$ satisfies $Q$, that $I_3$ also does.

  Let us now modify~$I_3$ using the second minimality criterion on~$e$ to dissociate the edges $e_4 = (u_1, v_1')$ and $e_4' = (u_1', v_1)$, simplifying the instance further. The weight of these edges is the same as that of~$e$, but their side weight is smaller: indeed, $u_1$ has exactly as many incident facts as~$u$ did in~$I_1$, and $v_1'$ has the same number as~$v$ in~$I_1$ except that $F_{\rr}$ is missing, so the side weight of~$e_4$ is indeed smaller. The same holds for~$e_4'$ because $v_1$ has exactly the same incident facts as~$v$ and $u_1'$ has the same as~$u$ except~$F_{\ll}$. This means that these edges are not tight, as otherwise it would contradict the second criterion in Definition~\ref{def:minimal}. Thus, we can dissociate one and then the other, and $Q$ will still be satisfied.
  Say we first dissociate~$e_4$: we create $u_5$ and $v_5'$ and replace $e_4$ by copies from $u_1$ to $v_5'$ and from $u_5$ to~$v_1'$, with $v_5'$ and $u_5$ being leaves. Next, we dissociate~$e_4'$, whose weight and side weight is unchanged relative to~$I_3$: and we create $u_5'$ and $v_5$ and replace $e_4'$ by copies from $u_1'$ to~$v_5$ and from $u_5'$ to~$v_1$, with $v_5$ and $u_5'$ being leaves. As we argued, the minimality of~$e$ ensures that the edges that we dissociate are not tight, so the resulting instance $I_4$ (see Figure~\ref{fig:magic}) still satisfies~$Q$.

  Now, we can finally merge back vertices to reach an instance $I_5$ isomorphic to the dissociation of~$e$ in~$I$. This will yield our contradiction, because we assumed that~$e$ is tight. Specifically, let us map $u_1'$ to~$u_1$ and $v_5$ to~$v_5'$: this defines a homomorphism because the edge $(u_1', v_5)$ can be mapped to $(u_1, v_5')$, this was the only edge involving $v_5$, and all other facts involving~$u_1'$ have a copy involving~$u_1$ by definition of the fine dissociation. Let us also map $v_1'$ to~$v_1$ and $v_5$ to~$v_5'$ in the same fashion, which defines a homomorphism for the same reason. The resulting instance $I_5$ (see Figure~\ref{fig:magic}) still satisfies~$Q$.
  Now observe that $I_5$ is isomorphic to the result of the (non-fine) dissociation of~$e$ in~$I$ (Figure~\ref{fig:dissociation}): we have added two leaves $u_5'$ and $v_5'$, the vertices $u_1$ and $v_1$ indeed correspond to~$u$ and~$v$, we have removed the edge from~$u$ to~$v$ and replaced it by copies from~$u_1$ to~$v_5'$ and from $u_5'$ to~$v_1$.

  Thus we have deduced that dissociating $e$ in~$I$ yields an instance that satisfies~$Q$. But as $(I, e)$ was a tight pattern, this is impossible, so we have reached a contradiction and the proof is finished.
\end{proof}

We can now conclude the proof of the main result of the section, Theorem~\ref{thm:ustcon}:
\begin{proof}[Proof of Theorem~\ref{thm:ustcon}]
  Fix the query $Q$ and the minimal tight pattern $(I, e)$. By definition, $e$ is then a non-leaf edge: pick an arbitrary incident pair $\Pi$ and a non-unary fact $F_{\mm}$ covered by~$e$. We show the \mbox{\#P-hardness} of $\PQE(Q)$ by reducing from U-ST-CON (Definition~\ref{def:stcon}).
        Given an st-graph~$G$, we apply the coding of Definition~\ref{def:concoding} and obtain a TID $\pdb$, which can be computed in polynomial time. As in the proof of Theorem~\ref{thm:pp2dnfred}, given a possible world~$\omega$ of~$G$, what matters is to show that (1.) if $\omega$ is good then $\phi(\omega)$ satisfies~$Q$, and (2.) if $\omega$ is bad then $\phi(\omega)$ violates~$Q$.

        For this, we use Proposition~\ref{prp:concoding}. For (1.), the result follows from the fact that the query~$Q$ is closed under homomorphisms, and the edge~$e$ was assumed to be iterable (Definition~\ref{def:iteration}), so it is iterable relative to any incident pair, in particular~$\Pi$. Thus, the iterates satisfy~$Q$, so $\phi(\omega)$ also does when~$\omega$ is good. For (2.), we know by Lemma~\ref{lem:magic} that the result of the fine dissociation does not satisfy~$Q$, so $\phi(\omega)$ does not satisfy it either when~$\omega$ is bad. This establishes the correctness of the reduction and concludes the proof.
\end{proof}

We have thus established Theorem~\ref{thm:main2}, and the main result of this paper.

\section{Generalizations of the Dichotomy Result}%
\label{sec:generalizations}

This section presents two generalizations of our main result. We first show that the dichotomy also applies to a special case of probabilistic query evaluation, known as \emph{generalized model counting}. Second, we strengthen the dichotomy result to the case where the signature can include  unary predicates in addition to binary predicates.

\subsection*{A special case of probabilistic query evaluation.}
 Recent work has studied the \emph{generalized (first-order) model counting problem~(GFOMC)}~\cite{KenigSuciu20}: given a TID $\Imc$ where $\Pr_\Imc(t) \in \{0, 0.5, 1\}$ for every tuple $t \in \Imc$, $\GFOMC(Q)$ for a query $Q$ is the problem of computing $\Pr_\Imc(Q)$. In other words, $\GFOMC(Q)$ is a special case of $\PQE(Q)$, where each atom $t$ in the TID can only have a probability $p \in \{0,0.5,1\}$.

To extend our result to this setting, we simply observe that all the reductions presented in this paper only use tuple  probabilities from $\{0,0.5,1\}$. Thus, all our hardness results for $\PQE(Q)$ thus immediately apply to $\GFOMC(Q)$ and we obtain a corollary to our main hardness result (Theorem~\ref{thm:main2}):
\begin{cor}%
  \label{cor:gmc}
Let $Q$ be an unbounded $\ucqinf$ query over an arity-two signature. Then, the
  problem $\GFOMC(Q)$ is {\normalfont\sharpP}-hard.
\end{cor}

One can then ask if our dichotomy (Theorem~\ref{thm:main}) also generalizes to the GFOMC problem. Clearly, if $\PQE(Q)$ can be computed in polynomial time, then so can $\GFOMC(Q)$, and hence, for any safe UCQ $Q$, the $\GFOMC(Q)$ problem is immediately in \FPTime by Dalvi and Suciu~\cite{dalvi2012dichotomy}. The other direction is more interesting, i.e., assuming a UCQ $Q$ is unsafe for $\PQE$, is it also unsafe for $\GFOMC$?
This was very recently shown to be true:
\begin{thm}[Theorem 2.2,~\cite{KenigSuciu20}]%
  \label{thm:suciu-gmc}
For any unsafe UCQ $Q$, $\GFOMC(Q)$ is {\normalfont\sharpP}-hard.
\end{thm}

In particular, this implies that all safe and unsafe queries coincide for UCQs, across the problems $\GFOMC$ and $\PQE$. Then, combining this theorem with Corollary~\ref{cor:gmc},
we can state our dichotomy result also for $\GFOMC$:

\begin{thm}[Dichotomy of \GFOMC]%
  \label{thm:gmc}
  Let $Q$ be a \ucqinf over an arity-two signature. Then, either~$Q$ is equivalent to a safe UCQ and $\GFOMC(Q)$ is in {\normalfont\FPTime}, or it is not and $\GFOMC(Q)$ is {\normalfont\sharpP}-hard.
\end{thm}
\begin{proof}
Let $Q$ be a safe UCQ\@. Then, since $\PQE(Q)$ is in \FPTime, so is $\GFOMC(Q)$. If $Q$ is not equivalent to a safe UCQ, then either it is equivalent to an unsafe UCQ and $\GFOMC(Q)$ is \sharpP-hard by Theorem 2.2~of~\cite{KenigSuciu20}, or it is an unbounded query in $\ucqinf$ and $\GFOMC(Q)$ is \sharpP-hard by Corollary~\ref{cor:gmc}.
\end{proof}

\subsection*{Allowing unary predicates.}
We now turn to the question of extending our results to support unary predicates. Recall that we claimed in Section~\ref{sec:prelim} that our results extend if the signature can feature unary and binary predicates. We now justify this claim formally by showing the analogue of Theorem~\ref{thm:main2} and Corollary~\ref{cor:gmc} for signatures with relations of arity 1 and~2.

\begin{thm}%
  \label{thm:main12}
  Let $Q$ be an unbounded \ucqinf over a signature with relations of arity 1 and~2. Then, $\GFOMC(Q)$, and hence $\PQE(Q)$, is \#P-hard.
\end{thm}

\begin{proof}
Fix the signature $\sigma$ and query $Q$. Let $\sigma'$ be the arity-two signature constructed from~$\sigma$ by replacing each relation $R$ of arity $1$ by a relation $R'$ of arity~$2$.
Considering $Q$ as an infinite union of CQs, we define $Q'$ as a \ucqinf on $\sigma'$ obtained by replacing every unary atom $R(x)$ in $Q$ with the atom $R'(x,x)$.
The resulting query $Q'$ is unbounded. Indeed, assume to the contrary that $Q'$ is equivalent to a UCQ $Q''$. As the truth of $Q'$ by construction only depends on the presence or absence of facts of the form $R'(a, a)$, not $R'(a, b)$ with $a \neq b$, we can assume that $Q''$ only contains atoms of the form $R'(x, x)$ and not $R'(x, y)$. Now,
replacing back each atom $R'(x,x)$ in $Q''$ with $R(x)$, we would obtain a UCQ that is equivalent to $Q$ over the signature $\sigma$, contradicting the unboundedness of $Q$.

Thus, as $Q'$ is an unbounded $\ucqinf$, we know by Theorem~\ref{thm:main2} and Corollary~\ref{cor:gmc} that $\GFOMC(Q')$ is \#P-hard.
Moreover, again by construction of~$Q'$, its satisfaction does not depend on the presence or absence of facts of the form $R'(a, b)$ with $a \neq b$.
This implies that $\GFOMC(Q')$ is \#P-hard even when assuming that the input TIDs contain no such facts.

Now, to show that $\GFOMC(Q)$ is \#P-hard, we reduce from $\GFOMC(Q')$ where input TIDs are restricted to satisfy this additional assumption. Consider such a TID $\calI' = (I', \pi')$. Consider the function $\phi$ that maps any instance $I$ over $\sigma'$ to the instance $\phi(I)$ obtained by replacing each fact $R'(a, a)$ by the fact $R(a)$.
We build in polynomial time the TID $\calI = (I, \pi)$ on $\sigma$, with $I \colonequals \phi(I')$, and with $\pi$ giving to each $\sigma$-fact of arity two in~$I$ the same probability as in~$I'$, and giving to each $\sigma$-fact $R(a)$ of arity 1 in~$I$ the probability of the fact $R'(a, a)$ in~$I'$.
Then, $\phi$ defines a probability-preserving bijection between the possible worlds of~$\Imc'$ and the possible worlds of~$\Imc'$, and by construction $\phi$ guarantees that a possible world of~$\Imc'$ satisfies $Q'$ iff its $\phi$ image satisfies~$Q$. This establishes that the reduction is correct, and concludes the proof.
\end{proof}

\section{Conclusions}%
\label{sec:conc}
We have shown that PQE is \#P-hard for any unbounded \ucqinf over an arity-two signature, and hence proved a dichotomy on PQE for all \ucqinf queries: either they are unbounded and PQE is \#P-hard, or they are bounded and the dichotomy by Dalvi and Suciu applies.
Our result captures many query languages; in particular disjunctive Datalog over binary signatures, regular path queries, and all ontology-mediated queries closed under homomorphisms.

There are three natural directions to extend our result.
First, we could study queries that are \emph{not} homomorphism-closed, e.g., with
disequalities or negation. We believe that this would require  different techniques as the problem is still open even when extending UCQs in this fashion (beyond the results of~\cite{FiOl16}).
Second, we could lift the arity restriction and work over signatures of arbitrary arity: we conjecture that
PQE is still \#P-hard for any unbounded \ucqinf in that case. Much of our proof techniques may adapt, but we do not know how to extend the definitions of dissociation, fine dissociation, and iteration. In particular, dissociation on a fact is difficult to adapt because incident facts over arbitrary arity signatures may intersect in complicated ways. We believe that the result could extend with a suitable dissociation notion and tight patterns with a more elaborate minimality criterion, but for now we leave the extension to arbitrary-arity signatures to future work.
Third, a natural question for future work is whether our hardness result on unbounded homomorphism-closed queries also applies to the \emph{(unweighted) model counting problem}, where all facts of the TID must have probability~$0.5$: the hardness of this problem has only been shown recently on the class of self-join free CQs~\cite{amarilli2020model} and on the so-called unsafe final type-I queries~\cite{KenigSuciu20}, but remains open as of this writing for unsafe UCQs in general.

\section*{Acknowledgments}
\noindent This work was supported by the UK EPSRC grant EP/R013667/1.

\bibliographystyle{alphaurl}
\bibliography{lib}

\end{document}